\documentclass[11pt,a4 paper]{article}
\newtheorem{theo}{Theorem}[section]

\newtheorem{rem}{Remark}[section]
\newtheorem{defi}{Definition}[section]
\newtheorem{lemma}{Lemma}[section]

\newtheorem{prop}{Proposition}[section]
\newtheorem{ex}{Example}[section]

\newtheorem{op}{Open Problem}
\newtheorem{construction}{Construction}

\newcommand{\proof}{\noindent{\em Proof.}\quad}

\usepackage{multirow,bm}
\usepackage{hhline}
\usepackage{array}
\usepackage{makecell}
\newcommand{\F}{\mathbb{F}}
\newcommand{\oo}{\overline}

\usepackage{float}
\usepackage{tocloft}
\usepackage{amsfonts,amsmath,amssymb}
\usepackage{multirow, verbatim}
\usepackage{shadow,enumerate}
\usepackage{makeidx}  
\usepackage{graphicx}
\usepackage{color}

\usepackage{mathrsfs}

\newcommand{\FB}{{\mathbb F}_{2}}
\def\whitebox{{\hbox{\hskip 1pt
        \vrule height 6pt depth 1.5pt
        \lower 1.5pt\vbox to 7.5pt{\hrule width
                  3.2pt\vfill\hrule width 3.2pt}%
        \vrule height 6pt depth 1.5pt
        \hskip 1pt } }}
\def\qed{\ifhmode\allowbreak\else\nobreak\fi\hfill\quad\nobreak\whitebox\medbreak}

\begin{document}
\title{Generic constructions of 5-valued spectra Boolean functions}

\author{
S. Hod\v zi\'c \footnote {University of Primorska, FAMNIT, Koper,
Slovenia, e-mail: samir.hodzic@famnit.upr.si}\and E.~Pasalic\footnote{
University of Primorska, FAMNIT \& IAM, Koper, Slovenia, e-mail: enes.pasalic6@gmail.com}
\and W. G. Zhang \footnote{ ISN Laboratory, Xidian University, Xi'an 710071, China,
e-mail: weiguozhang@vip.qq.com}
}
\date{}
\maketitle

\begin{abstract}
Whereas the design and properties of bent and plateaued functions have been frequently addressed during the past few decades, there are only a few  design methods of so-called 5-valued spectra Boolean functions whose Walsh spectra takes the values in $\{0, \pm 2^{\lambda_1}, \pm 2^{\lambda_2}\}$. Moreover, these design methods mainly regards the specification of these functions in their ANF (algebraic normal form) domain. In this article, we give a precise characterization of this class of functions in their spectral domain using the concept of a dual of plateaued functions. Both necessary and sufficient conditions on the Walsh support of these functions are given which then connects their design (in spectral domain) to a family of so-called {\em totally (non-overlap) disjoint spectra plateaued functions}. We identify some suitable families of plateaued functions having this property, thus providing some generic methods in the spectral domain. Furthermore, we also provide an extensive analysis of their constructions in the ANF domain and provide several generic design methods. The importance of this class of functions is manifolded, where apart from being suitable for some cryptographic applications we emphasize their property of being constituent functions in the so-called 4-bent decomposition.

\end{abstract}

\section{Introduction}

The concept of bent functions has been introduced by Rothaus \cite{Rot} as a family of Boolean functions possessing several nice combinatorial properties which allowed for their great range of applications such as design theory, coding theory, sequences, cryptography to mention a few. An exhaustive survey on bent functions related to their design and properties can be found in \cite{CarlMes2016}. Another related class of Boolean functions on $GF(2)^n$ whose Walsh spectra is three valued, thus $W_f(\omega) \in \{0, \pm 2^{r}\}$ for $r >\lceil n/2 \rceil$, is called a class of plateaued functions introduced by Zhang and Zheng \cite{Zheng}.
The design methods of plateaued functions have been addressed in several works in the past twenty years \cite{CarletAPN,CarletQ,Feng2,SihemPOP2014,THOMCusick2016,SihemM2017AMC,Fengrong2018IT} and their design in spectral domain recently in \cite{SHCF,SHEA}.

On the other hand, so called five-valued spectra Boolean functions is an interesting class of Boolean functions in its own. Being characterized by the property that their Walsh spectral values are in the set $\{0,\pm 2^{\lambda_1}, \pm 2^{\lambda_2}\}$ (where $\lceil \frac{n}{2} \rceil \leq \lambda_1 < \lambda_2 < n$) when considering an $n$-variable input space, these functions may satisfy multiple cryptographic criteria. Indeed, achieving the smallest possible $\lambda_i$ for odd $n$, namely having  $\lambda_1= \frac{n-1}{2}$ and  $\lambda_2= \frac{n+1}{2}$ ensures the same nonlinearity as those of semi-bent functions, and additionally their zero spectral values can be potentially suitably allocated so that they posses certain order of resiliency as well. Apart from their cryptographic applications, their design might contribute to a better understanding of other related combinatorial structures. In particular, one of the greatest problem related to some cryptographically significant functions is the existence of APN (almost perfect nonlinear) permutations for even $n \geq 8$. It is notable that the only known APN permutation over $\FB^6$ (up to equivalence), provided by Dillon \cite{Dillon}, has a 5-valued extended Walsh spectra (the spectra of all component functions) with the values $\{0,\pm 2^{\frac{n}{2}}, \pm 2^{\frac{n+2}{2}}\}$. Another important motivation regards the  so-called 4-bent decomposition \cite{Decom}.  In general,  a bent functions on $\F_2^n$ is then  viewed through  its (four) restrictions to the cosets  of some  $(n-2)$-dimensional linear subspace. These restrictions are then either bent, semi-bent or 5-valued spectra functions. Whereas the first two cases are well understood and addressed in the literature, the design of quadruples of 5-valued spectra functions suitable in the design of bent functions has not been investigated in detail.

To the best of our knowledge, the main existing research related to the  design of 5-valued spectra functions can be traced to the early work of Maitra and Sarkar \cite{Maitra} and some  recent articles \cite{CaoXu,XuCao,SihemM2017AMC}. However, all these methods address the design problem only partially, either by providing some simple methods in the ANF domain or by giving some sporadic classes  of 5-valued spectra functions through trace representation \cite{XuCao}. These methods are not generic and they do not give any useful insight that would contribute to a general design framework for this class of functions. In this article, we primarily address the notion of 5-valued spectra functions in the Walsh spectral domain which seems a natural approach to handle the design issue. More precisely, by requesting that the cardinalities of their Walsh supports w.r.t. to two different amplitudes are powers of two (thus $\#S_1(f)=\#\{\omega \in \FB^n: W_f(\omega) = \pm 2^{\lambda_1} \}=2^{r_1}$ and $\#S_2(f)=\#\{\theta \in \FB^n: W_f(\theta)\} = \pm 2^{\lambda_2} \}=2^{r_2}$) we relate their design to the design of plateaued functions. In general, the disjoint Walsh supports $S_1(f)$ and $S_2(f)$  induce the dual Boolean functions and the spectral design then relates to the following problem: Specify two disjoint subsets of $\FB^n$ of cardinality $2^{r_1}$ and $2^{r_2}$ respectively, and define the signs of the spectral values in $S_1(f)$ and $S_2(f)$ so that the resulting function is indeed a Boolean function! We provide both necessary and sufficient conditions regarding this specification which relies on a class of {\em totally disjoint spectra functions}. This notion is somewhat similar to a class of so-called {\em non-overlap disjoint spectra functions} introduced in \cite{DisjSpectra}  but is even more restrictive. We provide one generic solution to this problem by specifying a huge class of totally disjoint spectra functions that allows us to accomplish our goal.
This design approach is conducted entirely in the spectral domain and to the best of our knowledge this is a unique and novel design framework offering a broad generality and design flexibility.

In the second part of this article, we consider the possibility of specifying 5-valued spectra in the ANF domain. For this purpose, and again to provide a great variety of design methods, we employ so-called composite form (CF) representation of Boolean functions introduced recently in \cite{SHCF}. This representation is shown to be useful for deriving sufficient conditions on the initial function used in our construction methods,  so that the resulting functions are 5-valued spectra functions. We demonstrate that these conditions can be relatively easily satisfied and provide some generic examples (construction methods for every $n$). The advantage of this approach is that, in difference to the spectral design,  it gives  an explicit ANF form of a function without the need of getting it through the inverse Walsh transform. It should be remarked that  our design methods can be easily adopted for specifying four suitable 5-valued spectra functions that can be used to build bent functions. Due to a great variety of our design methods it is challenging to assume that the constructed bent functions do not necessarily belong to the known primary classes. The problem of confirming this exclusion is however difficult and needs to be considered separately.

 The rest of this article is organized as follows. In Section \ref{sec:pre}, we give some basic definitions related to Boolean functions and discuss the concept of dual of plateaued Boolean functions. In terms of the 4-bent decomposition we provide both necessary and sufficient conditions on 5-valued spectra restrictions of a bent functions in Section \ref{sec:dual+dec}. In Section \ref{sec:speccon}, we introduce the concept of totally disjoint spectra functions whose existence is shown to be both necessary and sufficient for designing 5-valued spectra in the spectral domain. One generic class of such function is also presented in this section, thus enabling a generic spectral design method. Some design methods in the ANF domain, based on the use of composite representation of Boolean functions are given in Section \ref{sec:CF}. In addition, the possibility of designing resilient 5-valued spectra functions using the so-called GMM method is addressed in Section \ref{sec:MM}. Some concluding remarks are given in Section \ref{sec:concl}.


\section{Preliminaries}\label{sec:pre}

%
The vector space $\mathbb{F}_2^n$ is the space of all $n$-tuples $x=(x_1,\ldots,x_n)$, where $x_i \in \mathbb{F}_2$.
For $x=(x_1,\ldots,x_n)$ and $y=(y_1,\ldots,y_n)$  in $\mathbb{F}^n_2$, the usual scalar (or dot) product over $\mathbb{F}_2$ is defined as $x\cdot y=x_1 y_1\oplus\cdots\oplus x_n y_n.$ The Hamming weight of  $x=(x_1,\ldots,x_n)\in \mathbb{F}^n_2$ is denoted and computed as  $wt(x)=\sum^n_{i=1} x_i.$ By "$\sum$" we denote the integer sum (without modulo evaluation), whereas "$\bigoplus$" denotes the sum evaluated modulo two.

The set of all Boolean functions in $n$ variables, which is the set of mappings from $\mathbb{F}_2^n$ to $\mathbb{F}_2$, is denoted by $\mathcal{B}_n$.  Especially, the set of affine functions in $n$ variables is given by $\mathcal{A}_n=\{a\cdot x\oplus b\;|\;a\in\mathbb{F}_2^n,\; b\in\{0,1\}\},$ and similarly  $\mathcal{L}_n=\{a\cdot x:a\in\mathbb{F}_2^n\}\subset \mathcal{A}_n$ 
 denotes the set of linear functions.
It is well-known that any $f:\mathbb{F}^n_2 \rightarrow \mathbb{F}_2$ can be uniquely represented by its associated algebraic normal form (ANF) as follows:
\begin{eqnarray}\label{ANF}
f(x_1,\ldots,x_n)={\bigoplus_{u\in \mathbb{F}^n_2}{\lambda_u}}{(\prod_{i=1}^n{x_i}^{u_i})},
\end{eqnarray}
where $x_i, \lambda_u \in \mathbb{F}_2$ and $u=(u_1, \ldots,u_n)\in \mathbb{F}^n_2$.

For an arbitrary function $f\in \mathcal{B}_n$, the set of its values on $\mathbb{F}^n_2$ (\emph{the truth table}) is defined as $T_f=(f(0,\ldots,0,0),f(0,\ldots,0,1),f(0,\ldots,1,0),\ldots,f(1,\ldots,1,1))$. The corresponding $(\pm 1)$-{\em sequence of $f$} is defined as
$\chi_f=((-1)^{f(0,\ldots,0,0)},(-1)^{f(0,\ldots,0,1)},(-1)^{f(0,\ldots,1,0)}\ldots,$ $(-1)^{f(1,\ldots,1,1)})$. The {\em Hamming distance} $d_H$ between two arbitrary Boolean functions, say $f,g\in \mathcal{B}_n,$ we define by $d_H(f,g)=\{x\in \mathbb{F}^n_2:f(x)\neq g(x)\}=2^{n-1}-\frac{1}{2}\chi_f\cdot \chi_g$, where $\chi_f\cdot \chi_g=\sum_{x\in \mathbb{F}^n_2}(-1)^{f(x)\oplus g(x)}$.


The \emph{Walsh-Hadamard transform} (WHT) of $f\in\mathcal{B}_n$, and its inverse WHT, at any point $\omega\in\mathbb{F}^n_2$ are defined, respectively,  by
\begin{eqnarray}\label{WHT}
W_{f}(\omega)=\sum_{x\in \mathbb{F}_2^n}(-1)^{f(x)\oplus \omega\cdot x},\;\;\;\;(-1)^{f(x)}=2^{-n}\sum_{\omega\in \mathbb{F}_2^n}W_f(\omega)(-1)^{\omega\cdot x}.
\end{eqnarray}


\subsection{Bent and plateaued functions and their duals}

Throughout this article we use the following definitions related to bent and plateaued functions:
\begin{itemize}
\item A function $f\in\mathcal{B}_n,$ for even $n$, is called {\em bent} if $W_f(u)=2^{\frac{n}{2}}(-1)^{f^*(u)}$
for a Boolean function $f^*\in \mathcal{B}_n$ which is also a bent function, called the {\it dual} of $f$.


\item Two functions $f$ and $g$ on $\FB^n$ are said to be at {\em bent distance} if $d_H(f,g)= 2^{n-1}\pm 2^{n/2-1}$. Similarly, for a subset $B\subset \mathcal{B}_n$, a function $f$ is said to be at bent distance to $B$ if for all $g\in B$ it holds that $d_H(f,g)=2^{n-1}\pm 2^{n/2-1}$.

\item A function $f\in \mathcal{B}_n$ is called {\em $s$-plateaued} if its Walsh spectrum only takes three values $0$ and $\pm 2^{\frac{n+s}{2}}$ (the value $2^{\frac{n+s}{2}}$ is called {\em amplitude}), where $s\geq 1$ if $n$ is odd and $s\geq 2$ if $n$ is even ($s$ and $n$ always have the same parity).
A class of $1$-plateaued functions for $n$ odd, or  $2$-plateaued for $n$ even, corresponds to so-called {\em semi-bent} functions.

\item The {\em Walsh support} of  $f\in \mathcal{B}_n$ is defined as $S_f=\{\omega\in \mathbb{F}^n_2\; :\; W_f(\omega)\neq0\}$ and for an $s$-plateaued function its cardinality is $\#S_f=2^{n-s}$ \cite[Proposition 4]{Decom}.


 \item
A {\em dual} function $f^*$ of an $s$-plateaued $f \in \mathcal{B}_n$ is defined through $W_f(\omega)=2^{\frac{n+s}{2}}(-1)^{f^*(\omega)},$ for $\omega\in S_f$.
     %
     %
To specify the dual function as $f^*:\FB^{n-s} \rightarrow \FB$ we use the concept of {\em lexicographic ordering}. That is, a subset  $E=\{e_0,\ldots,e_{2^{n-s}-1}\}\subset \mathbb{F}^{n}_2$ is ordered lexicographically if $|e_i| < |e_{i+1}|$ for any $i \in [0,2^{n-s}-2]$, where $|e_i|= \sum_{j=0}^{n-1}e_{i,n-1-j} 2^j$ denotes the integer representation of $e_i \in \mathbb{F}^n_2$.
Since $S_f$ is not ordered in general, {\em we will always represent it} as $S_f=v \oplus E$, where $E$ is lexicographically ordered for some fixed  $v \in S_f$ and $e_0=\textbf{0}_{n}$. For instance, if $S_f=\{(0,1,0),(0,1,1),(1,0,0), (1,0,1)\}$, by fixing $v=(0,1,1)\in S_f$, then $E=\{e_0,e_1,e_2,e_3\}=\{(0,0,0), (0,0,1),(1,1,0),(1,1,1)\}$ is ordered lexicographically and consequently $S_f$ is "ordered" as $S_f=\{\omega_0,\omega_1,\omega_2,\omega_3\}=\{(0,1,1),(0,1,0),(1,0,1),(1,0,0)\}$. 

 A direct correspondence between  $\FB^{n-s}$ and $S_f=\{\omega_0,\ldots,\omega_{2^{n-s}-1}\}$ is achieved through $E$ so that for lexicographically ordered $\FB^{n-s}=\{x_0, x_1, \ldots, x_{2^{n-s}-1}\}$ we have
\begin{eqnarray}\label{DPL}
 f^*(\omega_j)\leftrightsquigarrow f^*(e_j)\leftrightsquigarrow f^*(x_j),\;\;\; x_j \in \FB^{n-s}, \; e_j \in E,\;j\in[0,2^{n-s}-1].
 \end{eqnarray}
For the above example, the identification is given by Table \ref{tab1}.
\begin{table}[h!]
\scriptsize
\centering
\begin{tabular}{|c|c|c|}
  \hline
  \makecell{$S_f=v\oplus E$ given with respect to\\lexicographically ordered\\ $E$ w.r.t. $v=(0,1,1)$} & \makecell{The lexicographically\\ ordered set $E$} & \makecell{Values of $f^*(\omega_i)=f^*(e_i)=f^*(x_i)$,\\
  where $x_i\in \mathbb{F}^{2}_2$\\ $(\mathbb{F}^{2}_2$ is ordered lexicographically)} \\ \hline
  $\omega_0=(0,1,1)$ & $(0,0,0)=e_0$ & $f^*(\omega_0)=f^*(x_0)=f^*(0,0)$ \\
  $\omega_1=(0,1,0)$ & $(0,0,1)=e_1$ & $f^*(\omega_1)=f^*(x_1)=f^*(0,1)$ \\
  $\omega_2=(1,0,1)$ & $(1,1,0)=e_2$ & $f^*(\omega_2)=f^*(x_2)=f^*(1,0)$ \\
  $\omega_3=(1,0,0)$ & $(1,1,1)=e_3$ & $f^*(\omega_3)=f^*(x_3)=f^*(1,1)$ \\
  \hline
\end{tabular}
\caption{Defining $f^*:\FB^{n-s} \rightarrow \FB^2$ via support $S_f$.}
\label{tab1}
\end{table}
\end{itemize}

\section{5-valuated spectra functions in terms of bent decomposition}\label{sec:dual+dec}

In this section we define the dual of Boolean functions with $5$ values in their spectrum, as a natural extension of the concept of dual of plateaued functions (defined in Section \ref{sec:pre}). In addition, using the dual of 5-valued spectra functions we provide an alternative description of \cite[Theorem 7]{Decom} on decomposition of bent functions, which will be used in subsequent sections.

\subsection{On dual of 5-valuated spectra functions}\label{sec:dual5}

Let WHT spectrum of a function $f:\mathbb{F}^n_2\rightarrow \mathbb{F}_2$ contains the values $0,\pm c_1,\pm c_2$ ($c_1\neq c_2$), where $c_1,c_2\in \mathbb{N}$. For $i=1,2$, by $S^{[i]}_f\subset \mathbb{F}^n_2$ we denote the set $S^{[i]}_f=\{u\in \mathbb{F}^n_2:|W_f(u)|=c_i\}$, and  we define functions $f^*_{[i]}:S^{[i]}_f\rightarrow \mathbb{F}_2$ such that the following equality holds:
\begin{eqnarray}\label{eq:wdual5}
W_f(u)=\left\{
\begin{array}{cc}
 0, & u\not\in S^{[1]}_f\cup S^{[2]}_f, \\
 c_i\cdot (-1)^{f^*_{[i]}(u)}, & u\in S^{[i]}_f,\;\; i\in\{1,2\}\\
 \end{array}
\right..
\end{eqnarray}
Clearly, relation (\ref{eq:wdual5}) extends the definition of dual of a plateaued functions, since the functions $f^*_i$ are regulating the signs of integers $c_i$ in the spectrum (where $S^{[1]}_f\cap S^{[2]}_f=\emptyset$). Also note that the pairs $(S^{[1]}_f,f^*_{[1]})$ and $(S^{[2]}_f,f^*_{[2]})$ uniquely define the function $f$. Since throughout the article we will consider functions $f$ for which the sets $S^{[i]}_f$ are of the size $2^{\lambda_i}$ ($2^{\lambda_1}+2^{\lambda_2}<2^n$), we now provide the description of $f^*_{[i]}:S^{[i]}_f\rightarrow \mathbb{F}_2$ as functions from $\mathbb{F}^{\lambda_i}_2$ to $\mathbb{F}_2$.

For $i=1,2$, let $v_i\in \mathbb{F}^n_2$ and $E_i=\{e^{(i)}_0,\ldots,e^{(i)}_{2^{\lambda_i}-1}\}\subset \mathbb{F}^n_2$ ($e^{(i)}_0=\textbf{0}_n$) be lexicographically ordered subsets such that $S^{[i]}_f=\{\omega^{(i)}_0,\ldots,\omega^{(i)}_{2^{\lambda_i}-1}\}=v_i\oplus E_i$, where $\omega^{(i)}_j=v_i\oplus e^{(i)}_j$, for $j\in[0,2^{\lambda_i}-1]$. Clearly, the lexicographically ordered set $E_i$ imposes an ordering on $S^{[i]}_f$ with respect to equality $\omega^{(i)}_j=v_i\oplus e^{(i)}_j$. Using the representation of $S^{[i]}_f=v_i\oplus E_i$, the function $f^*_{[i]}$ as a mapping from $\mathbb{F}^{\lambda_i}_2$ to $\mathbb{F}_2$ is defined as
\begin{eqnarray}\label{eq:dp}
f^*_{[i]}(\omega^{(i)}_j)\leftrightsquigarrow f^*_{[i]}(e^{(i)}_j)\leftrightsquigarrow f^*_{[i]}(x_j),\;\;j\in[0,2^{\lambda_i}-1],\;\;i=1,2,
\end{eqnarray}
where $\mathbb{F}^{\lambda_i}_2=\{x_0,\ldots,x_{2^{\lambda_i}-1}\}$ is ordered lexicographically. In the next subsection we proceed with the analysis of decomposition of bent functions in terms of duals of its restrictions and corresponding Walsh supports.


\subsection{Decomposition of bent functions}\label{sec:platechar}

The decomposition of bent functions on $\FB^n$, $n$ is even, to affine subspaces  $a\oplus V$, for some $k$-dimensional linear subspace $V\subset \mathbb{F}^n_2$, was considered in \cite{Decom}.
For a bent function $\mathfrak{f} \in \mathcal{B}_n$, the restriction to $a\oplus V$ is denoted by $\mathfrak{f}_{a\oplus V}$ and it can be viewed as a function from $\FB^k \rightarrow \F_2$ using
\begin{eqnarray}\label{eq:rest}
\mathfrak{f}_{a\oplus V}(a\oplus v_i)\leftrightsquigarrow \mathfrak{f}_{a\oplus V}(v_i)\leftrightsquigarrow \mathfrak{f}_{a\oplus V}(x_i),\;\;\;i\in [0,2^k-1],
\end{eqnarray}
for lexicographically ordered $V=\{v_0,\ldots,v_{2^k-1}\}$ and $\mathbb{F}^k_2=\{x_0,\ldots,x_{2^k-1}\}$.
This identification between $V$ and $\mathbb{F}^k_2$, and thus the definition of $\mathfrak{f}_{a\oplus V}:\mathbb{F}^k_2\rightarrow \mathbb{F}_2$, strongly depends on the ordering of $V$. The reason why we use the lexicographic ordering above is  \cite[Lemma 3.1-$(ii)$]{SHCF} and any ordering of $V$ which satisfies the property \cite[Lemma 3.1-$(i)$]{SHCF} can be used instead.
\begin{defi}  Let $\mathfrak{f} \in \mathcal{B}_n$ and let $V$ be a linear subspace of $\mathbb{F}_2^n$
of dimension $k$. The decomposition of $\mathfrak{f}$ with respect to $V$ is the sequence $\{\mathfrak{f}_{a\oplus V}: a \in Q\}$ where $Q \oplus V=\mathbb{F}_2^n$ and all $\mathfrak{f}_{a\oplus V}$ are considered as Boolean functions in $\mathcal{B}_k$.
 \end{defi}
The 4-decomposition of a bent function $\mathfrak{f}\in \mathcal{B}_{n}$ then defines four subfunctions  on the four cosets  of some $(n-2)$-dimensional linear subspace \cite{Decom}. More precisely, for nonzero
$\alpha,\beta \in \mathbb{F}^n_2$ with $\alpha \neq \beta$ this $(n-2)$-dimensional subspace  is defined as $V=\langle \alpha,\beta \rangle^{\perp}$, where
 the {\em dual} of a linear subspace, say $S\subset \mathbb{F}^n_2$, is defined as $S^\perp=\{x\in \mathbb{F}^n_2: x\cdot y = 0,\;\forall y\in S\}.$

Let $(f_1,f_2,f_3,f_4)$ be such a decomposition, that is, $f_1, \ldots, f_4 \in \mathcal{B}_{n-2}$ are defined on the four cosets $\textbf{0}_n\oplus V, \alpha \oplus V, \beta \oplus V,(\alpha \oplus \beta) \oplus V$ respectively, thus $Q=\langle \alpha,\beta \rangle $
and $Q\oplus V=\mathbb{F}^n_2$ ($Q\cap V=\{\textbf{0}_n\}$).
Such a decomposition is called  a {\em bent 4-decomposition} when all $f_i$ ($i\in[1,4]$), are bent;  a {\em semi-bent 4-decomposition} when all $f_i$ ($i\in[1,4]$) are semi-bent;   a {\em 5-valued 4-decomposition} when all $f_i$ ($i\in[1,4]$) are 5-valued spectra functions so that $W_{f_i} \in \{0, \pm 2^{(n-2)/2}, \pm 2^{n/2}\} $ \cite{Decom}. These are the only possibilities and we strictly have that all the restrictions have the same spectral profile, for instance the restrictions cannot be a mixture of bent and semi-bent functions. The different cases that arise could be related to the second order derivatives (with respect to $\alpha$ and $\beta$) of the dual   $\mathfrak{f}^*$ of a bent function $ \mathfrak{f}$ \cite{Decom}.

Before we proceed further, we recall that following result which is actually \cite[Lemma 3.1-$(ii)$]{SHCF}.
\begin{lemma}\cite{SHCF}\label{lema1}
Let $S=\{\omega_{0},\ldots,\omega_{2^m-1}\}\subseteq \mathbb{F}^k_2$ be any  affine subspace of dimension $m\geq 2$ such that $S=v\oplus E$, for some lexicographically ordered linear subspace $E=\{e_0,\ldots,e_{2^m-1}\}\subseteq \mathbb{F}^k_2$ and $v\in S_f$, where $\omega_{i}=v\oplus e_i$ for $i\in [0,2^m-1]$. Then for an arbitrary vector $u\in \mathbb{F}^k_2$ it holds that
\begin{eqnarray}\label{hrow}
((-1)^{u\cdot \omega_{0}},(-1)^{u\cdot \omega_{1}},\ldots,(-1)^{u\cdot \omega_{2^m-1}})=(-1)^{\varepsilon_u} H^{(r_u)}_{2^{m}},
\end{eqnarray}
for some $0\leq r_u\leq 2^m-1$ and $\varepsilon_u\in\mathbb{F}_2$. In addition,  $\{T_{\ell}:\ell\in \mathcal{L}_m\}\subseteq \{(u\cdot e_0,\ldots,u\cdot e_{2^m-1}):u\in \mathbb{F}^k_2\}$, which means that  $\mathcal{L}_m$ is contained in a multi-set of $m$-variable linear functions whose truth tables are $\{(u\cdot e_0,\ldots,u\cdot e_{2^m-1}):u\in \mathbb{F}^k_2\}$.
\end{lemma}
Now we completely describe the 4-decomposition of a bent function in terms of Walsh supports and duals of its restrictions $f_{1},\ldots,f_4$.
\begin{theo}\label{th:dec}
Let $\mathfrak{f} \in \mathfrak{B}_n$ be a bent function, for even $n \geq 4$.  Let $\alpha,\beta \in {\mathbb{F}_2^n}^*$ $(\alpha\neq \beta)$ and  $V=\langle \alpha,\beta \rangle ^\perp$. If we denote by $(f_1,\ldots,f_4)$ the 4-decomposition of $\mathfrak{f}$ with respect to $V$, then $(f_1,\ldots,f_4)$ is:
\begin{enumerate}[i)]
\item A bent 4-decomposition if and only if it holds that $f^*_1\oplus f^*_2\oplus f^*_3\oplus f^*_4=1$.
\item A semi-bent 4-decomposition if and only if
functions $f_i$ $(i\in[1,4])$ are pairwise disjoint spectra  semi-bent functions. 
\item A five-valued 4-decomposition if and only if 
the following statements hold:
\begin{enumerate}[a)]
\item The sets $S^{[1]}_{f_i}=\{\vartheta\in \mathbb{F}^{n-2}_2:|W_{f_i}(\vartheta)|=2^{\frac{n}{2}}\}$ $(i\in[1,4])$ are pairwise disjoint;
\item All $S^{[2]}_{f_i}=\{\vartheta\in \mathbb{F}^{n-2}_2:|W_{f_i}(\vartheta)|=2^{\frac{n-2}{2}}\}$ are equal $(i\in[1,4])$, and for $f^*_{[2],i}:S^{[2]}_{f_i}\rightarrow \mathbb{F}_2$ it holds that $f^*_{[2],1}\oplus f^*_{[2],2}\oplus f^*_{[2],3}\oplus f^*_{[2],4}=1$. 
\end{enumerate}
\end{enumerate}
\end{theo}
\proof $i)$ Let the restrictions $f_i$ of $\mathfrak{f}$ be denoted by $\mathfrak{f}_q$ ($q\in Q=\{\textbf{0}_n,a,b,a\oplus b\}$), i.e., $(f_1,f_2,f_3,f_4)=(\mathfrak{f}_\textbf{0},\mathfrak{f}_a,\mathfrak{f}_b,\mathfrak{f}_{a\oplus b})$. We assume that  $V=\langle \alpha,\beta\rangle^{\perp}=\{v_0,\ldots,v_{2^{n-2}-1}\}$ is lexicographically ordered. Using that  $\mathbb{F}^n_2=\bigcup_{q\in Q}(q\oplus V)$ we compute the WHT of $\mathfrak{f}$, at arbitrary $u\in  \mathbb{F}^n_2$, as
%
\begin{eqnarray*}\label{equation}
W_\mathfrak{f}(u)=\sum_{q\in Q}(-1)^{u\cdot q}\sum_{v_i\in V}(-1)^{\mathfrak{f}(q\oplus v_i)\oplus u\cdot v_i}.
\end{eqnarray*}
By Lemma \ref{lema1}, we have that $\{u\cdot v:v\in V\}$ corresponds to the truth table of some linear function $\ell:\mathbb{F}^{n-2}_2\rightarrow \mathbb{F}_2$, say $\ell(x_1,\ldots,x_{n-2})=\vartheta_u\cdot x_j$ for some $\vartheta_u\in \mathbb{F}^{n-2}_2$ and  $x_j\in \mathbb{F}^{n-2}_2$, using the correspondence $u\cdot v_i\leftrightsquigarrow \vartheta_u\cdot x_i=\ell(x_i)$.

In addition, since $Q=\{\textbf{0}_n,a,b,a\oplus b\}$ is a linear subspace, then the sequence $\{(-1)^{u\cdot q}:q\in Q\}$ corresponds to some row of the Sylvester-Hadamard matrix of size $4\times 4$, i.e., for some integer $r_u\in[0,3]$ it holds that $(1,(-1)^{u\cdot a},(-1)^{u\cdot b},(-1)^{u\cdot (a\oplus b)})=H^{(r_u)}_4.$
 Consequently, since by (\ref{eq:rest}) we have that $\mathfrak{f}(q\oplus v_i)\leftrightsquigarrow \mathfrak{f}_q(x_i)$ ($i\in[0,2^{n-2}-1]$), then we can write
\begin{eqnarray*}\label{equation2}
W_\mathfrak{f}(u)&=&\sum_{q\in Q}(-1)^{u\cdot q}\sum_{x_i\in \mathbb{F}^{n-2}_2}(-1)^{\mathfrak{f}_q(x_i)\oplus \vartheta_u\cdot x_i}=\sum_{q\in Q}(-1)^{u\cdot q}W_{\mathfrak{f}_q}(\vartheta_u)\nonumber \\
&=&H^{(r_u)}_4\cdot(W_{f_1}(\vartheta_u),W_{f_2}(\vartheta_u),W_{f_3}(\vartheta_u),W_{f_4}(\vartheta_u)).
\end{eqnarray*}
Applying Lemma \ref{lema1} once again, when $u$ runs through the whole space $\mathbb{F}^n_2$ then $r_u$ runs through the whole set $\{0,1,2,3\}$ but also $\mathcal{L}_{n-2}\subset \{(u\cdot v_0,\ldots,u\cdot v_{2^{n-2}-1}):u\in \mathbb{F}^n_2\}$ (the latter set being a multi-set). In other words, $\mathfrak{f}$ is bent if and only if the system
\begin{eqnarray}\label{eq:system}
H_4\cdot (W_{f_1}(\vartheta),W_{f_2}(\vartheta),W_{f_3}(\vartheta),W_{f_4}(\vartheta))^T=(\pm 2^{\frac{n}{2}},\pm 2^{\frac{n}{2}},\pm 2^{\frac{n}{2}},\pm 2^{\frac{n}{2}})^T
\end{eqnarray}
holds for all vectors $\vartheta\in \mathbb{F}^{n-2}_2$. Now, if $f_i$ ($i\in[1,4]$) are bent on $\mathbb{F}^{n-2}_2$, that is $W_{f_i}(\vartheta)=2^{\frac{n-2}{2}}(-1)^{f^*_i(\vartheta)}$, then the system (\ref{eq:system}) holds if and only if $((-1)^{f^*_1(\vartheta)},\ldots,(-1)^{f^*_4(\vartheta)})$ is a sequence of a bent function in two variables, which is equivalent to $f^*_1\oplus f^*_2\oplus f^*_3\oplus f^*_4=1$.

$ii)$ Let $f_i$ be semi-bent on $\mathbb{F}^{n-2}_2$, for $i=1, \ldots, 4$.
The only possible solutions of the system (\ref{eq:system}) is such that if for some $\vartheta \in \F_2^{n-2}$ we have  $W_{f_i}(\vartheta)=\pm 2^{\frac{n}{2}}$ then necessarily  $W_{f_j}(\vartheta)=0$ for $j \neq i$. Since the Walsh supports $S_{f_i}$ partition the space $\mathbb{F}^{n-2}_2$, then $f_i$ have to be pairwise disjoint spectra functions.

$iii)$ In this case $W_{f_i} \in \{0, \pm 2^{(n-2)/2}, \pm 2^{n/2}\} $ \cite{Decom}. Clearly, if $\mathfrak{f}$ is bent on $\FB$ then  both a) and b) are necessary.
Then, similarly as in $i)$, we have that $f^*_{[2],1}\oplus f^*_{[2],2}\oplus f^*_{[2],3}\oplus f^*_{[2],4}=1$, where $f^*_{[2],i}:S^{[2]}_{f_i}\rightarrow \mathbb{F}_2$ and $S^{[2]}_{f_i}=\{\vartheta\in \mathbb{F}^{n-2}_2:|W_{f_i}(\vartheta)|=2^{\frac{n-2}{2}}\}$ $(i\in[1,4])$.\qed
Theorem \ref{th:dec}-$(i)$ means that the concatenation of bent functions $f_1,\ldots,f_4:\mathbb{F}^{n-2}_2\rightarrow \mathbb{F}_2$ gives a bent function in $n$ variables if and only if  $f^*_1\oplus f^*_2\oplus f^*_3\oplus f^*_4=1$ holds. However, note that this fact is known as \cite[Theorem 7]{Pren} (or \cite[Theorem 1]{Tokareva}). Some further comments on  Theorem \ref{th:dec}-$(ii)$ will be given in Section \ref{sec:speccon}.
\begin{rem}\label{remm1}
The construction of bent functions $f_1,\ldots,f_4$ with the property that $f^*_1\oplus f^*_2\oplus f^*_3\oplus f^*_4=1$ holds (where $f_4=f_1\oplus f_2\oplus f_3$) has been addressed in \cite{Nastja}. 
Note that the Theorem \ref{th:dec}-$(i)$ does not require that $f_4=f_1\oplus f_2\oplus f_3$, and thus the condition $f^*_1\oplus f^*_2\oplus f^*_3\oplus f^*_4=1$ is much easier to satisfy.
\end{rem}
In the rest of the article, we focus on construction of 5-valued spectra functions which satisfy the properties of Theorem \ref{th:dec}-$(iii)$.

\section{Spectral construction of 5-valued spectra functions}\label{sec:speccon}

The so-called spectral technique of constructing plateaued functions (which include semi-bent functions), which precisely specifies the Walsh support and the signs of the nonzero Walsh coefficients, has been introduced in \cite{SHCF} and later slightly extended in \cite{SHEA}. In this section, we recall this spectral method (used to construct plateaued functions)  and extend it to the case of 5-valued spectra functions.

For an arbitrary $s$-plateaued function $f$ defined on $\mathbb{F}^n_2$ let $S_f\subset \mathbb{F}^n_2$ ($\#S_f=2^{n-s}$) be its Walsh support ordered as $S_f=\{\omega_{0},\ldots,\omega_{2^{n-s}-1}\}$. The \emph{sequence profile} of $S_f$,
 which is a multi-set of $2^n$ sequences of length $2^{n-s}$ induced by $S_f$, is defined as 
\begin{eqnarray}\label{RSf}
\Phi_{f}=\{\phi_u:\mathbb{F}^{n-s}_2\rightarrow \mathbb{F}_2\;:\; \chi_{\phi_u}=((-1)^{u\cdot \omega_{0}},(-1)^{u\cdot \omega_{1}},\ldots,(-1)^{u\cdot \omega_{2^{n-s}-1}}),\;\omega_i\in S_f,\; u\in \mathbb{F}^{n}_2\}.
\end{eqnarray}
As noted in \cite{SHCF}, $\Phi_{f}$ depends on the ordering of $S_f$ and it is spanned by the functions $\phi_{b_1},\ldots,\phi_{b_n}$, i.e., $\Phi_{f}=\langle \phi_{b_1},\ldots,\phi_{b_n}\rangle$, where $b_1,\ldots,b_n$ is the canonical basis of $\mathbb{F}^n_2$ ($b_i$ contains the non-zero coordinate at the $i$-th position).

The following result  summarizes the spectral method described in \cite[Section 3]{SHCF}.
\begin{theo}\cite{SHCF}\label{th:plateH}
Let  $S_f=v\oplus E=\{\omega_0,\ldots,\omega_{2^{n-s}-1}\} \subset \FB^n$, with $\#S_f=2^{n-s}$, for some $v \in \mathbb{F}^n_2$ and lexicographically ordered subset $E=\{e_0,e_1, \ldots, e_{2^{n-s}-1}\}\subset \mathbb{F}^n_2$ with $e_0=\textbf{0}_n$ $($where $\omega_i=v\oplus e_i)$. For a function $g:\FB^{n-s}\rightarrow \mathbb{F}_2$ with $wt(g)=2^{n-s-1}\pm 2^{\frac{n-s}{2}-1}$, let the Walsh spectrum of $f$ be defined (by identifying $x_i \in \FB^{n-s}$ and $e_i \in E$) as
\begin{eqnarray}\label{eq:walshX}
W_f(u)= \left \{  \begin{array}{ll}
  2^{\frac{n+s}{2}}(-1)^{g(x_i)} & \textnormal{ for } u= v \oplus e_i \in S_f,\\
 0 & u \not \in S_f. \\
\end{array}  \right.
\end{eqnarray}
Then:
\begin{enumerate}[i)]
\item $f$ is an $s$-plateaued function if and only if $g$ is at bent distance to $\Phi_{f}$ defined by (\ref{RSf}).
\item If $E\subset \mathbb{F}^n_2$ is a linear subspace, then $f$ is $s$-plateaued  if and only if $g \in \mathcal{B}_{n-s}$ is bent.
\end{enumerate}
\end{theo}
\begin{rem}
The case when $S_f$ is represented as $S_f=v\oplus EM$, where $M$ is an invertible binary matrix, is given by \cite[Theorem 3.3]{SHEA}.
\end{rem}
In the context of Theorem \ref{th:dec}-$(ii)$ which regards the  design of $f_1,\ldots,f_4$ suitable for a semi-bent 4-decomposition, the Walsh supports $S_{f_i}$ are necessarily pairwise disjoint.
Nevertheless, Theorem \ref{th:plateH}-$(i)$ requires that one must ensure that
  $f^*_i$ is at bent distance to $\Phi_{f_i}$, for any $i\in [1,4]$. The latter condition is however easily satisfied if the Walsh supports $S_{f_i}$ are the cosets of a linear subspace, since in that case the set $\Phi_{f_i}$ contains only linear/affine functions (due to Lemma \ref{lema1}). The selection of dual  bent functions $f^*_i:\mathbb{F}^{n-4}_2\rightarrow \mathbb{F}_2$  of $f_i$ in relation (\ref{eq:walshX}) is then arbitrary.

In order to extend this approach, thus  to design  5-valued spectra function described by Theorem \ref{th:dec}-$(iii)$ in the spectral domain, we further clarify the essence of this approach.
Namely, if we want to design  an $s$-plateaued function $f \in \mathcal{B}_n$ with Walsh support $S_f=v\oplus E$ ($v\in \mathbb{F}^n_2$, $E\subset \mathbb{F}^n_2$, $\#S_f=2^{n-s}$), then at arbitrary $u\in \mathbb{F}^n_2$ (using the inverse WHT)  we have:
\begin{eqnarray}\label{eq:pspec}
\sum_{\omega_i\in S_f}(-1)^{f^*(\omega_i)\oplus u\cdot \omega_i}=\sum_{x_i\in \mathbb{F}^{n-s}_2}(-1)^{f^*(x_i)\oplus \phi_u(x_i)}=\chi_{f^*}\cdot \chi_{\phi_u}=2^{\frac{n-s}{2}}(-1)^{f(u)},
\end{eqnarray}
where $\phi_u(x_i)\leftrightsquigarrow \phi_u(\omega_i)=u\cdot\omega_i$, $i\in[0,2^{n-s}-1]$. Thus, for a given $S_f$ the main problem is to specify the dual function $f^*$ which specifies the signs of the nonzero Walsh coefficients so that $\chi_{f^*}\cdot \chi_{\phi_u}=\pm 2^{\frac{n-s}{2}}$, for all $u\in \mathbb{F}^n_2$. In other words, $f^*$ and $\phi_u$ must be at bent distance for any $u\in \mathbb{F}^n_2$. If this condition is not satisfied, then $f$ is not a Boolean function.


The spectral design  of 5-valued spectra functions essentially relies on some strict properties of the dual functions that need to be fulfilled.
\begin{prop}\label{prop:spec}
Let  $n$ be even and specify  $S^{[1]}_f=\{u\in \mathbb{F}^n_2:|W_f(u)|=2^{\frac{n+2}{2}}\}$, $S^{[2]}_f=\{u\in \mathbb{F}^n_2:|W_f(u)|=2^{\frac{n}{2}}\}$ with $S^{[1]}_f\cap S^{[2]}_f=\emptyset$ $(\#S^{[1]}_f+ \#S^{[2]}_f<2^n)$. Assume  that the spectrum $W_f$ is constructed as
\begin{eqnarray}\label{eq:wX}
W_f(u)= \left\{\begin{array}{cc}
                 0, & u\not\in S_f \\
                 (-1)^{f^*_{[1]}(u)}\cdot 2^{\frac{n+2}{2}}, & u\in S^{[1]}_f \\
                 (-1)^{f^*_{[2]}(u)}\cdot 2^{\frac{n}{2}}, & u\in S^{[2]}_f
               \end{array}
\right.,
\end{eqnarray}
where the duals $f^*_{[i]}:S^{[i]}_{f}\rightarrow \mathbb{F}_2$ $(i=1,2)$. Then, $W_f=(W_f(u_0),\ldots, W_f(u_{2^n-1}))$ is a spectrum of a Boolean function $f:\mathbb{F}^n_2\rightarrow \mathbb{F}_2$ if and only if the equality
\begin{eqnarray}\label{speceq}
2X_1(u)+X_2(u)=(-1)^{\varepsilon_u} 2^{\frac{n}{2}},
\end{eqnarray}
holds for all $u\in \mathbb{F}^n_2$ $(\varepsilon_u\in \{0,1\})$, where $X_i(u)=\sum_{\omega\in S^{[i]}_f}(-1)^{f^*_{[i]}(\omega)\oplus u\cdot \omega}$, $i=1,2$.
\end{prop}
\proof For arbitrary $u\in \mathbb{F}^n_2$, by the inverse WHT formula (\ref{WHT}) we have that
$$2^n(-1)^{f(u)}=\sum_{\omega\in \mathbb{F}^n_2}W_f(u)(-1)^{u\cdot \omega}=2^{\frac{n+2}{2}}\sum_{\omega\in S^{[1]}_f}(-1)^{f^*_{[1]}(\omega)\oplus u\cdot \omega}+2^{\frac{n}{2}}\sum_{\omega\in S^{[2]}_f}(-1)^{f^*_{[2]}(\omega)\oplus u\cdot \omega},$$
which is equivalent to $2^{\frac{n}{2}}(-1)^{f(u)}=2X_1(u)+X_2(u)$, i.e., the statement holds.\qed

This result, through equation (\ref{speceq}), provides both necessary and sufficient conditions for which $W_f$ constructed by means of (\ref{eq:wX}) is 5-valued spectrum of a Boolean function.
In order to solve equation (\ref{speceq}) for unknown pairs $(S^{[i]}_f,f^*_{[i]})$ ($i=1,2$), we  generalize the notion of non-overlap disjoint spectra property introduced in \cite{DisjSpectra} by adding an additional constraint relevant to our context.
\begin{defi}\label{def:gsp}
For two disjoint sets $S^{[1]}_f, S^{[2]}_f\subset \mathbb{F}^{n}_2$, with  $\#S^{[1]}_f + \# S^{[2]}_f= 2^{\lambda_1}+2^{\lambda_2}< 2^n,$ we say that functions $f^*_{[1]}:S^{[1]}_f\rightarrow \mathbb{F}_2$ and $f^*_{[2]}:S^{[2]}_f\rightarrow \mathbb{F}_2$  are {\em totally  disjoint spectra functions} if it holds that $$X_1(u)X_2(u)=0\;\;\;\text{and}\;\;\;|X_1(u)|+|X_2(u)|>0$$ for all $u\in \mathbb{F}^n_2$, where $X_i(u)=\sum_{\omega\in S^{[i]}_f}(-1)^{f^*_{[i]}(\omega)\oplus u\cdot \omega}$, $i=1,2$.
\end{defi}
\begin{rem}
Note that the second condition implies the nonexistence of a vector $u\in \mathbb{F}^n_2$ for which $X_1(u)=X_2(u)=0$. This prevents from getting a contradiction in  (\ref{speceq}) since the right side $\pm 2^{\frac{n}{2}}$ in (\ref{speceq}) must not be equal to zero. Without this condition the notion of totally disjoint spectra coincides with non-overlap disjoint spectra functions in \cite{DisjSpectra}.
\end{rem}
The following result connects the amplitudes and dimensions of the duals $f_1^*$ and $f_2^*$ so that 5-valued spectra functions can be derived from totally disjoint spectra functions.
\begin{prop}\label{prop:con1}
Let $S^{[i]}_f=\{\omega^{(i)}_0,\ldots,\omega^{(i)}_{2^{\lambda_i}-1}\}=v\oplus E_i\subset\mathbb{F}^n_2$ $(n$ even$)$ be disjoint affine subspaces, where $E_i=\{e^{(i)}_0,\ldots,e^{(i)}_{2^{\lambda_i}-1}\}$ $(i=1,2$, $e^{(i)}_0=\textbf{0}_n)$ and $2^{\lambda_1}+2^{\lambda_2}<2^n$. Suppose that $f^*_{[i]}:\mathbb{F}^{\lambda_i}_2\rightarrow\mathbb{F}_2$ are totally disjoint spectra $s_i$-plateaued functions. Then, $W_f=(W_f(u_0),\ldots,W_f(u_{2^{n}-1}))$ constructed by (\ref{eq:wX}) is a spectrum of a 5-valued spectra function $f:\mathbb{F}^n_2\rightarrow \mathbb{F}_2$ if and only if $\lambda_1+s_1+2=\lambda_2+s_2=n$.
\end{prop}
\proof We only need to prove that under the given conditions the equation (\ref{speceq}) is satisfied for all $u\in \mathbb{F}^n_2$. By Lemma \ref{lema1}, we have that $(u\cdot \omega^{(i)}_0,\ldots,u\cdot \omega^{(i)}_{2^{\lambda_i}-1})$ are truth tables of linear/affine functions defined on $\mathbb{F}^{\lambda_i}_2$. Consequently, the equality in (\ref{speceq}) is equivalent to
\begin{eqnarray*}\label{speceq2}
2W_{f^*_{[1]}}(\vartheta_u)+W_{f^*_{[2]}}(\theta_u)=\pm 2^{\frac{n}{2}},\;\;\;\vartheta_u\in \mathbb{F}^{\lambda_1}_2,\;\;\theta_u\in \mathbb{F}^{\lambda_2}_2,
\end{eqnarray*}
where $\{u\cdot \omega^{(1)}_j:\omega^{(1)}_j\in S^{[1]}_f\}=\{\vartheta_u\cdot x_j:x_j\in \mathbb{F}^{\lambda_1}_2\}$ and $\{u\cdot \omega^{(2)}_j:\omega^{(2)}_j\in S^{[2]}_f\}=\{\theta_u\cdot x_j :x_j\in \mathbb{F}^{\lambda_2}_2\}$. Using the fact that $f^*_{[i]}$ are totally disjoint spectra functions, we have the following:
\begin{eqnarray}\label{speceq2}
\pm 2^{\frac{n}{2}}=2W_{f^*_{[1]}}(\vartheta_u)+W_{f^*_{[2]}}(\theta_u)=\left\{\begin{array}{cc}
                                                                                   \pm 2\cdot 2^{\frac{\lambda_1+s_1}{2}}, & \vartheta_u\in S_{f^*_{[1]}}\subset \mathbb{F}^{\lambda_1}_2 \\
                                                                                   \pm 2^{\frac{\lambda_2+s_2}{2}}, & \theta_u\in S_{f^*_{[2]}}\subset \mathbb{F}^{\lambda_2}_2
                                                                                 \end{array}
\right..
\end{eqnarray}
Note that the totally disjoint spectra property for plateaued functions $f^*_{[i]}$ means that either $\vartheta_u$ belongs to the Walsh support $S_{f^*_{[1]}}$ and $\theta_u$ does not belong to $S_{f^*_{[2]}}$, or vice versa (due to the second property of Definition \ref{def:gsp}). Now, the equality (\ref{speceq2}) holds if and only if $\lambda_1+s_1+2=\lambda_2+s_2=n$ holds, and thus a function $f:\mathbb{F}^n_2\rightarrow \mathbb{F}_2$ obtained from the constructed spectrum (applying inverse WHT) is a 5-valued spectra function.\qed

The following example illustrates the possibility of specifying totally disjoint spectra functions.
\begin{ex}\label{ex:spec1}
Let $E_1,E_2\subset \mathbb{F}^6_2$ be linear subspaces given as
$$E_1=\{e^{(1)}_0,\ldots,e^{(1)}_{7}\}=\{(0,0,0)\}\times \mathbb{F}^3_2,\;\;\;E_2=\{e^{(2)}_0,\ldots,e^{(2)}_{31}\}=\{0\}\times \mathbb{F}^5_2.$$
Also, for $v_1=\bf{0}_6$ and $v_2=(1, 0, 0, 1, 0, 0)$, let $S^{[1]}_f=E_1$ and $S^{[2]}_f=v_2\oplus E_1$ $(S^{[1]}_f\cap S^{[2]}_f=\emptyset)$.\\
 Let us define two plateaued functions
$$f^*_{[1]}(x_1,x_2,x_3)=x_1x_2,\;\;\;\; f^*_{[2]}(x_1,x_2,x_3,x_4,x_5)=x_1 x_2 \oplus x_3 x_4\oplus x_5$$
on lexicographically ordered sets $E_i$ using  $f^*_{[1]}(v_1\oplus e^{(1)}_i)\leftrightsquigarrow f^*_{[1]}(y_i)$ $(y_i\in \mathbb{F}^3_2$, $i\in[0,7])$ and  $f^*_{[2]}(v_2\oplus e^{(2)}_j)\leftrightsquigarrow f^*_{[2]}(w_j)$ $(w_j\in \mathbb{F}^{5}_2$, $j\in[0,31])$. It  can be easily verified that $f^*_{[1]}$ and $f^*_{[2]}$ are totally disjoint spectra functions. More precisely, both being 1-plateaued functions on corresponding domains, we have that
$$(X_1(u),X_2(u))\in \{(\pm 4,0), (0,\pm 8)\}$$
holds for all $u\in \mathbb{F}^6_2$, where $X_i(u)=\sum_{e\in E_i}(-1)^{f^*_{[i]}(v_i\oplus e)\oplus u\cdot (v_i\oplus e)}$.
The spectrum of  $f:\mathbb{F}^6_2\rightarrow \mathbb{F}_2$, whose Walsh supports are $S^{[i]}_f$ and duals $f^*_{[i]}$, is given  as (using (\ref{eq:wX}))
\begin{eqnarray*}
W_f=\hskip -2mm &(&\hskip -3mm -16, -16, -16, -16, -16, -16, 16, 16, 0, 0, 0, 0, 0, 0, 0, 0, 0, 0,
0, 0, 0, 0, 0, 0, 0, 0, 0, 0, 0, 0,\\
&&\hskip -5mm 0, 0, -8, 8, 8, -8, -8, 8, -8, 8,
-8, 8, 8, -8, -8, 8, -8, 8, -8, 8, 8, -8, -8, 8, -8, 8, 8, -8,\\
&&\hskip -5mm -8, 8, 8, -8, 8, -8).
\end{eqnarray*}
By applying the inverse WHT to $W_f$ one recovers the ANF of $f:\mathbb{F}^6_2\rightarrow \mathbb{F}_2$ which is given as
$$f(x_1,\ldots,x_6)=1 \oplus x_1 x_6 \oplus x_2 x_3 x_6 \oplus x_4 (x_5 \oplus x_6).$$
\end{ex}

\subsection{Specifying totally disjoint spectra functions}

The design of totally disjoint spectra plateaued functions is apparently a harder task than constructing regular disjoint spectra functions (using Theorem \ref{th:plateH} for instance), due to the facts that $f^*_{[i]}$ are defined on different spaces $\mathbb{F}^{\lambda_i}_2$ and sums $X_i$ depend on vectors $u\in \mathbb{F}^n_2$. The previous example illustrates their existence but nevertheless a generic methods for their construction is needed. The following generic approach utilize bent functions as initial functions for this purpose.

\begin{construction}
Let $n$, $m$ and $k$ be even with $n=m+k$. Let $h\in \mathcal{B}_m$ and $g\in \mathcal{B}_k$ be two bent functions.
Let $H$ be any subspace of $\F^m_2$ of co-dimension 1, and let $\oo{H}=\F^m_2\backslash H$.
Let also $E_1=\F^k_2\times H$ and $E_2=\{\textbf{0}_k\}\times \oo{H}$.
We construct the spectra of $f\in \mathcal{B}_n$ as follows:
\begin{equation}\label{Wf}
  W_{f}(\alpha, \beta)=
  \begin{cases}
  (-1)^{g(\alpha)\oplus h(\beta)}\cdot 2^{n/2},& (\alpha,\beta)\in E_1\\
  (-1)^{h(\beta)}\cdot 2^{m/2+k},& (\alpha,\beta)\in E_2\\
  0,& otherwise.
  \end{cases}
\end{equation}
Then $W_f$ is the spectra of a Boolean function $f\in \mathcal{B}_n$.
Let now
\begin{eqnarray*}
  f_1(\alpha,\beta) &=& g(\alpha)\oplus h(\beta),~~(\alpha,\beta)\in E_1\\
  f_2(\alpha,\beta) &=& h(\beta),~~(\alpha,\beta)\in E_2.
\end{eqnarray*}
Then $f_1$ and $f_2$ are totally disjoint spectra functions.
\end{construction}

\begin{proof}
Let $x\in \F^k_2$ and $y\in \F^m_2$. For any $(\alpha,\beta)\in\F^n_2$, by the inverse WHT formula, we have
\begin{eqnarray*}
  2^{n}(-1)^{f(x,y)} &=&\sum_{(\alpha,\beta)\in \F^n_2} W_f(\alpha,\beta)(-1)^{(x,y)\cdot (\alpha,\beta)} \\
    &=& 2^{n/2}\sum_{(\alpha,\beta)\in E_1}(-1)^{g(\alpha)\oplus h(\beta)}(-1)^{(x,y)\cdot (\alpha,\beta)}
     +  2^{m/2+k}\sum_{(\alpha,\beta)\in E_2} (-1)^{h^*(\beta)} (-1)^{(x,y)\cdot (\alpha,\beta)}\\
    &=& 2^{n/2}\sum_{\alpha\in \F^k_2}(-1)^{g(\alpha)\oplus x\cdot \alpha}\sum_{\beta\in H}(-1)^{h(\beta)\oplus y\cdot\beta}
     +  2^{m/2+k}\sum_{\alpha=\textbf{0}_k}(-1)^{x\cdot \alpha}\sum_{\beta\in \oo{H}}(-1)^{h(\beta)\oplus y\cdot\beta}\\
    &=& 2^{n/2}W_g(x)\sum_{\beta\in H}(-1)^{h(\beta)\oplus y\cdot\beta}
     +  2^{m/2+k}\sum_{\beta\in \oo{H}}(-1)^{h(\beta)\oplus y\cdot\beta}\\
    &=& \pm 2^{m/2+k}\sum_{\beta\in H}(-1)^{h(\beta)\oplus y\cdot\beta}
     +  2^{m/2+k}\sum_{\beta\in \oo{H}}(-1)^{h(\beta)\oplus y\cdot\beta}\\
\end{eqnarray*}
Note that $H\cup \oo{H}=\F^m_2$ and $H\cap \oo{H}=\emptyset$. We always have
 \begin{equation}\label{H}
   \sum_{\beta\in H}(-1)^{h(\beta)\oplus y\cdot\beta}\cdot\sum_{\beta\in \oo{H}}(-1)^{h(\beta)\oplus y\cdot\beta}=0
 \end{equation}
 and
 \begin{equation}\label{cH}
     \sum_{\beta\in H}(-1)^{h(\beta)\oplus y\cdot\beta}+\sum_{\beta\in \oo{H}}(-1)^{h(\beta)\oplus y\cdot\beta}=\pm 2^{m/2},
 \end{equation}
 since the restrictions of $h$ to $H$ and $\overline{H}$ are disjoint spectra semi-bent functions in $m-1$ variables \cite[Theorem V.3]{Pascale}. Consequently, this implies
 \begin{equation*}\label{}
 2^{n}(-1)^{f(x,y)}=\pm 2^{m/2+k}\sum_{\beta\in H}(-1)^{h(\beta)\oplus y\cdot\beta}
     +  2^{m/2+k}\sum_{\beta\in \oo{H}}(-1)^{h(\beta)\oplus y\cdot\beta}=\pm 2^n
 \end{equation*}
 always holds. This proves $f$ is a Boolean function. By (\ref{H}) and (\ref{cH}), $f_1$ and $f_2$ are obviously totally disjoint spectra functions.
\end{proof}
\begin{rem}
 When $k=2$, this construction satisfies  the equation (\ref{speceq}) in Proposition \ref{prop:spec}, since the amplitudes in (\ref{Wf}) are exactly $2^{\frac{n}{2}}$ and $2^{\frac{n+2}{2}}$. 
\end{rem}

\begin{op}
Provide another generic construction methods of totally disjoint spectra plateaued functions for which the duals $f^*_{[i]}$ and/or Walsh supports $S_{f^*_{[i]}}$ can be fixed in advance.
\end{op}

\section{ANF-constructions of 5-valued spectra functions}\label{sec:CF}

In this section we mainly employ the so-called composite form (CF)-representation \cite{SHCF} of Boolean functions for the purpose of deriving secondary constructions of 5-valued spectra functions satisfying  the conditions imposed by Theorem \ref{th:dec}-$(iii)$. 


Let $\mathfrak{f}:\mathbb{F}^n_2 \rightarrow \mathbb{F}_2$ be an arbitrary Boolean function given in the $CF$-representation as
\begin{eqnarray*}\label{F2}
\mathfrak{f}(x)=f(H(x))=f(h_1(x),\ldots,h_k(x)),
\end{eqnarray*}
with the form $f:\mathbb{F}^k_2 \rightarrow \mathbb{F}_2$ and vectorial function $H=(h_1,\ldots,h_k):\mathbb{F}^n_2 \rightarrow \mathbb{F}^k_2$ which is not unique for a given form $f$.
 %
%
The WHT of the function $\mathfrak{f}$ is given by \cite[Proposition 9.1]{CarletBoolean}  as
\begin{eqnarray}\label{mainF}
W_\mathfrak{f}(u)=\sum_{x\in \mathbb{F}^n_2}(-1)^{f(h_1(x),\ldots,h_k(x))\oplus u\cdot x}=2^{-k}\sum_{\omega\in \mathbb{F}^k_2}W_f(\omega)W_{\omega\cdot (h_1,\ldots,h_k)}(u),\;\;\;u\in \mathbb{F}^n_2.
\end{eqnarray}
We mainly consider  plateaued forms due to the fact that their spectral constructions is essentially given by Theorem \ref{th:plateH}. In particular, if $f$ is an $s$-plateaued function in $k$ variables it implies that $W_{\mathfrak{f}}$ in (\ref{mainF}) can be written as
\begin{eqnarray}\label{mainF2}
W_\mathfrak{f}(u)=2^{\frac{s-k}{2}}\sum_{\omega\in S_f}(-1)^{f^*(\omega)}W_{\omega\cdot (h_1,\ldots,h_k)}(u),\;\;\;u\in \mathbb{F}^n_2.
\end{eqnarray}

\subsection{Secondary constructions using disjoint variable spaces}\label{sec:sepvar}

We follow the notation from \cite{SHCF} and represent the Walsh support $S_f$ of a plateaued function  $f:\mathbb{F}^k_2\rightarrow \mathbb{F}_2$ as $S_f=\Delta\wr \Theta$, where $\Delta$ is the set of the first $t$ ($<k$) coordinates of vectors $\omega\in S_f\subseteq \mathbb{F}^k_2$ and $\Theta$ is the set of the remaining $m=k-t$ coordinates of $\omega.$ More precisely, an arbitrary vector $\omega=(\omega_1,\ldots,\omega_t,\omega_{t+1},\ldots,\omega_k)\in S_f$  will be written as $\omega=(\delta,\theta)\in \Delta\wr \Theta=S_f$, where $\delta=(\omega_1,\ldots,\omega_t)\in \Delta$ and $\theta=(\omega_{t+1},\ldots,\omega_k)\in \Theta.$

Additionally, assuming that $\Theta$ is not a multi-set, for an arbitrary vector $\omega\in S_f=\Delta\wr \Theta$  written as $\omega=(\delta,\theta)$, by $\vartheta_{\omega}: \Theta\rightarrow \Delta$ we denote the function which maps $\theta$ to $\delta$, hence $\vartheta_{\omega}(\theta)=\delta$ (or simply $\vartheta(\theta)=\delta$ if it is clear that $(\delta,\theta)\in S_f$). 

Now we slightly extend \cite[Lemma 4.1]{SHCF} and recall Theorem 4.2-$(iii)$ given in \cite{SHCF} which is important for our main goal.
\begin{lemma}\label{rotlemma}
Let $\mathfrak{f}:\mathbb{F}^r_2\times\mathbb{F}^m_2\rightarrow \mathbb{F}_2$ be given as $\mathfrak{f}(x,y)=f(H(x,y))$, where the Walsh support of $f:\mathbb{F}^k_2\rightarrow \mathbb{F}_2$ can be written as  $S_f=\Delta\wr \Theta$ ($m=k-t$, $t \geq1$). Let $H(x,y)=(h_1(x,y),\ldots,h_k(x,y)):\mathbb{F}^r_2\times\mathbb{F}^m_2\rightarrow \mathbb{F}^k_2$ be a vectorial function such that
\begin{eqnarray}\label{hi}
\left\{\begin{array}{cc}
    h_i(x,y)=h_i(x), & i=1,\ldots,t,\;\; x\in \mathbb{F}^r_2, \\
    (h_{t+1}(x,y),\ldots,h_k(x,y))=(y_1,\ldots,y_m)=y\in  \mathbb{F}^m_2, & t+m= k.
  \end{array}\right.
\end{eqnarray}
Then, for any  $(u,v)\in \mathbb{F}^r_2\times\mathbb{F}^m_2$ the WHT of  $\mathfrak{f}=f(h_1,\ldots,h_k)$ is given by
\begin{eqnarray}\label{eq:wal}
W_\mathfrak{f}(u,v)=\left\{\begin{array}{cc}
                             2^{-t}\sum_{(\delta,v)\in S_f=\Delta\wr \Theta}W_{f}(\delta,v)W_{\delta\cdot (h_{1},\ldots,h_t)}(u), & v\in \Theta \\
                             0, & v\not\in \Theta
                           \end{array}
\right..
\end{eqnarray}
\end{lemma}
\begin{rem}\label{remlema}
 If $\mathfrak{f}=a\oplus d(h_1,\ldots,h_k)$ with $a(x,y)=a(x)$, then in (\ref{eq:wal}) instead of $W_{\delta\cdot (h_1,\ldots,h_t)}(u)$ we have $W_{a\oplus\delta\cdot (h_1,\ldots,h_t)}(u)$. Note that in \cite[Lemma 4.1]{SHCF} (in comparison to Lemma \ref{rotlemma}) we have that $\Theta=\mathbb{F}^m_2$.
\end{rem}
\begin{theo}\cite{SHCF}\label{P1}
Let $\mathfrak{f}:\mathbb{F}^r_2\times\mathbb{F}^m_2\rightarrow \mathbb{F}_2$ be given as $\mathfrak{f}(x,y)=f(H(x,y))=f(h_{1}(x),\ldots,h_s(x),y),$ where $f:\mathbb{F}^k_2\rightarrow \mathbb{F}_2$ is $s$-plateaued  and $H=(h_1,\ldots,h_k)$ is a vectorial function defined by (\ref{hi}) ($t=s$). Assume that $S_f=\Delta\wr \Theta$  with $\Theta=\mathbb{F}^m_2$ and ${\bf 0}_s\not\in \Delta$ ($m\geq 2$ is even, $s+m=k$). 
%

\noindent If for every $\delta\in \Delta$ it holds that  $\delta\cdot (h_1,\ldots,h_s)$ is $c_{\delta}$-plateaued  with (possibly) different amplitudes $2^{\frac{r+c_{\delta}}{2}}$, then  $W_\mathfrak{f}(\omega) \in \{0,\pm 2^{\frac{r+m+c_{\delta}}{2}}:\delta\in \Delta\}$. 
\end{theo}
\begin{rem}
If in Theorem \ref{P1} the functions $\delta\cdot (h_1,\ldots,h_s)$ (where $r,m$ are even) are bent or semi-bent on $\mathbb{F}^r_2$ (thus having that $c_{\delta}\in\{0,2\}$), then $\mathfrak{f}$ is a 5-valued spectra function. In general, this result provides  a generic method of constructing functions which are not necessarily 5-valued spectra functions. 
\end{rem}
%
The following result is a straightforward analysis when  a vectorial function $H:\FB^r \times \FB^m \rightarrow \FB$ is given as $H(x,y)=(a(x),h_1(x),h_2(x),y_1,y_2)$, which corresponds to the case when $r=s=3$ and $m=2$.
\begin{theo}\label{th:C1}
Let $\mathfrak{f}:\mathbb{F}^r_2\times\mathbb{F}^2_2\rightarrow \mathbb{F}_2$ $(n=r+2$, $r$ even$)$ be given as
$$\mathfrak{f}(x,y)=a(x)\oplus (h_1(x) \oplus y_1) (h_2(x) \oplus y_2),\;\;\;x\in \mathbb{F}^r_2,\;\;y=(y_1,y_2)\in \mathbb{F}^2_2.$$
If any of the following conditions hold:
$$\left\{\begin{array}{l}
           a\;\;\text{is bent},\;\;g=h_1=h_2\;\;\text{and}\;\;a\oplus g\;\;\text{is semi-bent}, \\
           a\;\;\text{is bent},\;\;h_1\in \mathcal{A}_r\;\;\text{and}\;\;a\oplus h_2\;\;\text{is semi-bent}, \\
           a\;\;\text{has 5-valued spectra}\;\;\text{and}\;\;h_1,h_2\in \mathcal{A}_r, \\
            a\;\;\text{has 5-valued spectra}\;\;\text{and}\;\;h_1\in \mathcal{A}_r\;\;\text{and}\;\;a\oplus h_2\;\;\text{is (semi-)bent}, \\
         \end{array}
\right.$$
 then $\mathfrak{f}$ is 5-valued spectra function and  $W_{\mathfrak{f}}(u,v)\in \{0,\pm 2^{\frac{n}{2}},\pm 2^{\frac{n+2}{2}}\},$ for $(u,v)\in \mathbb{F}^r_2\times\mathbb{F}^2_2.$
\end{theo}
\proof One may verify that the form $f(x_1,\ldots,x_5)=x_1 \oplus (x_2 \oplus x_4) (x_3 \oplus x_5)$ of  $\mathfrak{f}$ is a $3$-plateaued function in $k=5$ variables whose Walsh support is $$S_f=\Delta\wr \mathbb{F}^2_2=(\{1\}\times \mathbb{F}^2_2)\wr \mathbb{F}^2_2=\{(1, 0, 0, 0, 0), (1, 0, 1, 0, 1),(1, 1, 0, 1, 0), (1, 1, 1, 1, 1)\},$$ and its dual is $\chi_{f^*}=(1,1,1,-1)$. Since $\mathfrak{f}(x,y_1,y_2)=f(a(x),h_1(x),h_2(x),y_1,y_2)$, then by Lemma \ref{rotlemma} and Remark \ref{remlema} the WHT of $\mathfrak{f}$ at arbitrary $(u,v)\in \mathbb{F}^r_2\times  \mathbb{F}^2_2$ is given as:
\begin{eqnarray*}
W_{\mathfrak{f}}(u,v)=2\hskip -3mm\sum_{(\delta,v)\in \Delta\wr \mathbb{F}^2_2}(-1)^{f^*(\delta,v)}W_{\delta\cdot (a,h_1,h_2)}(u)=
\left\{\begin{array}{lc}
2(-1)^{f^*(1,0,0;0,0)}W_a(u), & v=(0,0) \\
2(-1)^{f^*(1,0,1;0,1)}W_{a\oplus h_2}(u), & v=(0,1) \\
2(-1)^{f^*(1,1,0;1,0)}W_{a\oplus h_1}(u), & v=(1,0)  \\
2(-1)^{f^*(1,1,1;1,1)}W_{a\oplus h_1\oplus h_2}(u), & v=(1,1)
\end{array}\right..
\end{eqnarray*}
Under the given conditions  the statement easily follows. \qed

\begin{rem}\label{rem:g}
Notice that the requirement on $a\oplus g$ to be semi-bent ($a$  is bent or 5-valued spectra function) can be easily satisfied since there are no conditions imposed on $g$. Namely, one can take any semi-bent function $d$ on $\mathbb{F}^r_2$, and then just take the function $g=a\oplus d$. 
\end{rem}
This method can be efficiently employed to provide 5-valued spectra functions on $\FB^n$ which can  be used to define the  restrictions  of a bent function on $\FB^{n+2}$ in terms of the 4-bent decomposition, cf. Theorem \ref{th:dec}-$(iii)$.
\begin{theo}\label{th:C2}
Let functions $\mathfrak{f}_1,\ldots,\mathfrak{f}_4:\mathbb{F}^r_2\times\mathbb{F}^2_2\rightarrow \mathbb{F}_2$ $(n=r+2$, $r\geq 4$ even$)$ be given as
$$\mathfrak{f}_i(x,y)=a_i(x)\oplus (g_i(x) \oplus y_1) (g_i(x) \oplus y_2),\;\;\;x\in \mathbb{F}^r_2,\;\;y=(y_1,y_2)\in \mathbb{F}^2_2.$$
If $a_i$ are bent functions such that $a^*_1\oplus a^*_2\oplus a^*_3\oplus a^*_4=1$, and $a_i\oplus g_i$ are pairwise disjoint spectra  semi-bent functions on $\mathbb{F}^r_2$, then $\mathfrak{f}_i$  satisfy the condition of Theorem \ref{th:dec}-$(iii)$.
\end{theo}
\proof By Theorem \ref{th:C1}, $\mathfrak{f}_i$ are 5-valued spectra functions, since $a_i$  and $a_i\oplus g_i$ ($g_i=h_1=h_2$) are bent respectively semi-bent on $\mathbb{F}^r_2$. From the proof of Theorem \ref{th:C1} (considering the values of the dual $f^*$) the WHT of $\mathfrak{f}_i$ is given by
\begin{eqnarray*}
W_{\mathfrak{f}_i}(u,v)=
\left\{\begin{array}{lc}
\pm 2W_{a_i}(u), & v\in\{(0,0),(1,1)\} \\
2W_{a_i\oplus g_i}(u), & v\in\{(0,1),(1,0)\}
\end{array}\right.,\;\;\;(u,v)\in \mathbb{F}^r_2\times \mathbb{F}^2_2.
\end{eqnarray*}
Note that the signs "$\pm$" of $2W_{a_i}(u)$ are exactly $(-1)^{f^*(\delta,v)}$ for $\delta\in\{(1,0,0),(1,1,1)\}$.
Since $a_i$ is bent then $|W_{\mathfrak{f}_i}(u,v)|=2^{\frac{n}{2}}$, for $(u,v)\in \mathbb{F}^r_2\times \{(0,0),(1,1)\}.$ Using the notation of Theorem \ref{th:dec}, we have that $S^{[2]}_{\mathfrak{f}_i}=\mathbb{F}^r_2\times \{(0,0),(1,1)\}$, and thus the dual $\mathfrak{f}^*_{[2],i}:S^{[2]}_{\mathfrak{f}}\rightarrow \mathbb{F}_2$ is given as $$\mathfrak{f}^*_{[2],i}(u,v)=a^*_i(u)\oplus f^*(\vartheta(v),v),$$ where $v\in\{(0,0),(1,1)\}.$ Recall that $\vartheta(v)=\delta\in\{(1,0,0),(1,1,1)\}$ for which it holds that $(\delta,v)\in S_f$ (and the vector $\vartheta(v)$ is unique due to the structure of $S_f=(\{1\}\times \mathbb{F}^2_2)\wr \mathbb{F}^2_2$). Since for a fixed vector $v$ we have that $f^*(\vartheta(v),v)$ is fixed (for all $i\in[1,4]$), then $a^*_1\oplus a^*_2\oplus a^*_3\oplus a^*_4=1$ clearly implies that $\mathfrak{f}^*_{[2],1}\oplus \mathfrak{f}^*_{[2],2}\oplus  \mathfrak{f}^*_{[2],3}\oplus  \mathfrak{f}^*_{[2],4}=1$ holds, which is the second condition of Theorem \ref{th:dec}-$(iii)$.

Using the fact that $a_i\oplus g_i$ is semi-bent, we have that $|W_{\mathfrak{f}_i}(u,v)|=2^{\frac{n+2}{2}}$ exactly for $(u,v)\in S_{a_i\oplus g_i}\times \{(0,0),(1,1)\}$, which gives that $S^{[1]}_{\mathfrak{f}_i}=S_{a_i\oplus g_i}\times \{(0,0),(1,1)\}$ ($S_{a_i\oplus g_i}$ is the Walsh support of $a_i\oplus g_i$). Clearly, for vectors $(u,v)\not\in S_{a_i\oplus g_i}\times \{(0,0),(1,1)\}$ we have that $W_{\mathfrak{f}_i}(u,v)=0$. Having the pairwise disjoint spectra  property for $a_i\oplus g_i$ implies that all $S^{[1]}_{\mathfrak{f}_i}$ are pairwise disjoint sets, which means that the first condition of Theorem \ref{th:dec}-$(iii)$ holds.\qed

\begin{ex}
Disjoint spectra semi-bent functions $a_i\oplus g_i:\mathbb{F}^r_2 \rightarrow \FB^2$ can be easily constructed by means of  Theorem \ref{th:plateH} by specifying pairwise disjoint spectra semi-bent functions $d_i:\mathbb{F}^r_2\rightarrow \mathbb{F}_2$ $(i\in[1,4]$, $\#S_{d_i}=2^{r-2})$ using four disjoint Walsh supports each of size $2^{r-2}$. Since there are no special conditions on functions $g_i$ in Theorem \ref{th:C2} (Remark \ref{rem:g}), we simply define $g_i=a_i\oplus d_i$, and thus the functions $a_i\oplus g_i=d_i$ trivially satisfy conditions of Theorem \ref{th:C2}.
\end{ex}
\begin{rem}
Notice that for $n=6$ the extended Walsh spectra of the Dillon's APN permutation $F$ is exactly 5-valued of the form $W_{F_u} \in \{0,\pm 2^{\frac{n}{2}}, \pm 2^{\frac{n+2}{2}}\}$, where $F_u=u\cdot F$ denotes the component function of $F$ chosen by $u \in {\FB^n}^*$. In terms of the amplitudes, this spectra exactly coincide with the spectra of functions designed by means of Theorem \ref{th:C2}.
\end{rem}

\subsection{Employing plateaued functions with non-affine Walsh support}
It is well-known that plateaued functions whose Walsh support is some affine subspace are equivalent to partially bent functions.  Their design in the spectral domain has been recently addressed in \cite{SHEA} and moreover some design methods using non-affine supports (being a more difficult task) have been proposed.   To simplify the necessary conditions for the purpose of designing 5-valued spectra functions, we  use specific plateaued and relatively simple  non-affine Walsh support of special kind.

Assume that  the Walsh support of some 2-plateaued function $f$ on $\FB^6$ corresponds to $S_f=I_{4\times 4}\times \mathbb{F}^2_2$, where $I_{4\times 4}$ is the identity matrix, so that $S_f$ is not affine subspace.
If we order $S_f$ as $S_f=(I_{4\times 4}\times (0,0))\cup \ldots \cup (I_{4\times 4}\times (1,1))$,
then viewing $S_f$ as a matrix of size $16\times 6$  the columns of $S_f$ correspond to truth tables of functions $\phi_{b_i}:\mathbb{F}^4_2\rightarrow \mathbb{F}_2$ ($i\in[1,6]$, $b_i$ has non-zero coordinate at $i$-th position) given by
$$(\phi_{b_1}(x),\ldots,\phi_{b_6}(x))=((1 \oplus x_3) (1 \oplus x_4), (1 \oplus x_3) x_4, x_3 (1 \oplus x_4), x_3x_4, x_1, x_2).$$
Then, by Theorem \ref{th:plateH}, $f$ will be a Boolean function (and then necessarily 2-plateaued) if one can define a dual function $f^*$ (with respect to a given ordering of $S_f$) which is at bent distance to $\Phi_f=\langle \phi_{b_1}, \ldots , \phi_{b_6} \rangle$.
Taking $f^*(x_1,\ldots,x_4)=(x_1,x_2)\cdot (x_3,x_4)$, one can verify that $f^*$ is at bent distance to $\Phi_f$.
Thus, the form $f:\mathbb{F}^6_2\rightarrow \mathbb{F}_2$ constructed by means of $(S_f,f^*)$ is a $2$-plateaued function given by
\begin{eqnarray}\label{eq:formath1}
f(x_1,\ldots,x_6)=(x_2 \oplus x_2 x_5 \oplus x_4 x_5) x_6 \oplus x_3 x_5 (1 \oplus x_6) \oplus x_1 (1 \oplus x_5) (1 \oplus x_6).
\end{eqnarray}
\begin{theo}\label{th:C3}
Let $\mathfrak{f}:\mathbb{F}^r_2\times\mathbb{F}^2_2\rightarrow \mathbb{F}_2$ $(n=r+2$, $r\geq 6$ even$)$ be given in the form $f:\mathbb{F}^6_2\rightarrow \mathbb{F}_2$, defined by (\ref{eq:formath1}), as $\mathfrak{f}(x,y)=f(h_1(x),\ldots,h_4(x),y_1,y_2),$ $x\in \mathbb{F}^r_2,$ $y=(y_1,y_2)\in \mathbb{F}^2_2.$ Then:
\begin{enumerate}[(i)]
\item If $h_i:\mathbb{F}^r_2\rightarrow \mathbb{F}_2$ are pairwise disjoint spectra  plateaued functions such that $h_1,h_2,h_3$ are 2-plateaued and $h_4$ is $4$-plateaued, then $W_\mathfrak{f}(u,v)\in \{0,\pm 2^{\frac{n}{2}},\pm 2^{\frac{n+2}{2}}\}$, $(u,v)\in \mathbb{F}^r_2\times\mathbb{F}^2_2.$
\item Let  $\mathfrak{f}_1,\ldots,\mathfrak{f}_4:\mathbb{F}^r_2\times\mathbb{F}^2_2\rightarrow \mathbb{F}_2$ be given as $\mathfrak{f}_i(x,y)=f(h_{1,i}(x),\ldots,h_{4,i}(x),y)$, where $h_{p,i}$ satisfy properties as in $(i)$, for  $i,p\in[1,4]$. 
    Additionally, let $h_{4,i}$ be pairwise disjoint spectra functions  for all $i\in[1,4]$, and assume that
 $$\left\{\begin{array}{l}
               S_{h_{p,1}}=S_{h_{p,2}}=S_{h_{p,3}}=S_{h_{p,4}}, \\
               h^*_{p,1}\oplus h^*_{p,2}\oplus h^*_{p,3}\oplus h^*_{p,4}=1,\;\;\;\forall p=1,2,3.
             \end{array}
    \right.$$
Then the functions $\mathfrak{f}_i$   satisfy the properties of Theorem \ref{th:dec}-$(iii)$, thus $(\mathfrak{f}_1,\ldots, \mathfrak{f}_4)$ is a 5-valued 4-decomposition of a bent function.
\end{enumerate}
\end{theo}
\proof $i)$ Since $S_f=I_{4\times 4}\times\mathbb{F}^2_2$, by Lemma \ref{rotlemma}, at arbitrary $(u,v)\in \mathbb{F}^r_2\times \mathbb{F}^2_2$, the WHT of $\mathfrak{f}$ ($t=4$, $s=2$) is given by
\begin{eqnarray*}
W_{\mathfrak{f}}(u,v)=\sum_{(\delta,v)\in S_f=I_{4\times 4}\times \Theta} (-1)^{f^*(\delta,v)}W_{\delta\cdot (h_1,\ldots,h_4)}(u).
\end{eqnarray*}
Due to the structure of the form (\ref{eq:formath1}), one can notice that for a fixed $v\in \mathbb{F}^2_2$  we have that the values $\{(-1)^{f^*(\delta,v)}:\delta\in I_{4\times 4}\}$ constitute the rows of the Sylvester-Hadamard matrix.  This is because   $f^*(x_1,\ldots,x_4)=(x_1,x_2)\oplus (x_3,x_4)$ is a concatenation of linear functions (with respect to given ordering of $S_f=I_{4\times 4}\times \mathbb{F}^2_2$ mentioned earlier).  Consequently, we have that 
\begin{eqnarray*}
\left(
  \begin{array}{c}
    W_{\mathfrak{f}}(u,00) \\
    W_{\mathfrak{f}}(u,01) \\
    W_{\mathfrak{f}}(u,10) \\
    W_{\mathfrak{f}}(u,11) \\
  \end{array}
\right)=\left(
          \begin{array}{cccc}
            1 & 1  & 1 & 1 \\
            1 & -1 & 1 & -1 \\
            1 & 1 & -1 & -1 \\
            1 & -1 & -1 & 1 \\
          \end{array}
        \right)\left(
                 \begin{array}{c}
                   W_{h_1}(u) \\
                  W_{h_2}(u) \\
                   W_{h_3}(u) \\
                   W_{h_4}(u) \\
                 \end{array}
               \right)= H_4\left(
                 \begin{array}{c}
                   W_{h_1}(u) \\
                  W_{h_2}(u) \\
                   W_{h_3}(u) \\
                   W_{h_4}(u) \\
                 \end{array}
               \right).
\end{eqnarray*}
Now, since $h_1,h_2,h_3$ are $2$-plateaued ($\#S_{h_i}=2^{r-2}$, $i=1,2,3$) and $h_4$ is $4$-plateaued ($\#S_{h_4}=2^{r-4}$), then clearly $S_{h_i}$ do not partition the space $\mathbb{F}^r_2$ and thus there exist vectors $(u,v)\in \mathbb{F}^r_2\times \mathbb{F}^2_2$ such that $W_{\mathfrak{f}}(u,v)=0$. In general, we have that $W_{\mathfrak{f}}(u,v)$ is precisely given as
$$W_{\mathfrak{f}}(u,v)=\left\{\begin{array}{cc}
 0, & u\not\in S_{h_i},\;\;\forall i\in[1,4],\\
                                 W_{h_i}(u)=2^{\frac{r+2}{2}}(-1)^{h^*_i(u)}, & u\in S_{h_i},\;\;i\in{1,2,3}, \\
                                 W_{h_4}(u)=2^{\frac{r+4}{2}}(-1)^{h^*_4(u)}, & u\in S_{h_4},
                               \end{array}
\right.\;\;\;\forall v\in \mathbb{F}^2_2.$$
Since $n=r+2$, we clearly have that $\mathfrak{f}$ is 5-valued spectra function with $W_\mathfrak{f}(u,v)\in \{0,\pm 2^{\frac{n}{2}},\pm 2^{\frac{n+2}{2}}\}$.

$ii)$ From the previous computation we have that
$$S^{[1]}_{\mathfrak{f}_i}=\{(u,v)\in \mathbb{F}^r_2\rightarrow \mathbb{F}_2: |W_{\mathfrak{f}_i}(u,v)|=2^{\frac{n+2}{2}}\}=S_{h_{4,i}}\times \mathbb{F}^2_2,$$
$$S^{[2]}_{\mathfrak{f}_i}=\{(u,v)\in \mathbb{F}^r_2\rightarrow \mathbb{F}_2: |W_{\mathfrak{f}_i}(u,v)|=2^{\frac{n}{2}}\}=\bigcup_{p=1,2,3} (S_{h_{p,i}}\times \mathbb{F}^2_2).$$
Since $S_{h_{4,i}}$ and $S_{h_{4,j}}$ are pairwise disjoint  for $i \neq j$, then apparently $S^{[1]}_{\mathfrak{f}_i}$ are pairwise disjoint which implies that the first condition of Theorem \ref{th:dec}-$(iii)$ holds. Furthermore, since $S_{h_{p,1}}=S_{h_{p,2}}=S_{h_{p,3}}=S_{h_{p,4}}$ holds for all $p=1,2,3$, then $S^{[2]}_{\mathfrak{f}_i}$ $(i\in[1,4])$ are  equal as sets. Also, the dual $\mathfrak{f}^*_{[2],i}:S^{[2]}_{\mathfrak{f}_i} \rightarrow \mathbb{F}_2$ is given by $\mathfrak{f}^*_{[2],i}(u,v)=h^*_{p,i}(u)$ for $u\in S_{h_{p,i}}$ ($v\in \mathbb{F}^2_2$), where $p=1,2,3$. Consequently, the set of equalities $$h^*_{p,1}\oplus h^*_{p,2}\oplus h^*_{p,3}\oplus h^*_{p,4}=1; \;\; \textnormal{ for } p=1,2,3, $$ implies that $\mathfrak{f}^*_{[2],1}\oplus \mathfrak{f}^*_{[2],2}\oplus \mathfrak{f}^*_{[2],3}\oplus \mathfrak{f}^*_{[2],4}=1$, which is the second condition of Theorem \ref{th:dec}-$(iii)$.\qed

\begin{rem}
The conditions of Theorem \ref{th:C3} can be easily satisfied if we use Theorem \ref{th:plateH}. More precisely, one can design pairwise disjoint spectra functions by taking disjoint  Walsh supports which  are affine subspaces. Then, the  duals have to be bent where additionally  $h^*_{p,1}\oplus h^*_{p,2}\oplus h^*_{p,3}\oplus h^*_{p,4}=1$ holds. 
\end{rem}
So far  our  results have been  based on Theorem \ref{P1}, that is when $\Theta=\mathbb{F}^m_2$ and some of the coordinate functions $h_i$ are linear (see relation (\ref{hi})). 
In what follows, we derive new constructions with non-linear coordinate functions (with disjoint variables) which are based on ideas present in \cite[Section 4.3]{SHCF}. A plateaued  form is constructed using the Walsh support and dual given by
$$S_f=\left(
        \begin{array}{ccccc}
          1 & 0 & 0 & 0 & 1  \\
          1 & 0 & 1 & 0 & 1  \\
          1 & 1 & 0 & 1 & 0  \\
          1 & 1 & 1 & 1 & 0  \\
        \end{array}
      \right),\;\;\;f^*(x_1,x_2)=x_1x_2,
$$
which using  Theorem \ref{th:plateH} gives that  $f:\mathbb{F}^5_2\rightarrow \mathbb{F}_2$  is a $3$-plateaued function given by
\begin{eqnarray}\label{eq:formth22}
f(x_1,\ldots,x_5)=x_1 \oplus x_5 \oplus x_3 (x_2 \oplus x_4 \oplus x_5).
\end{eqnarray}
\begin{theo}\label{th:C4}
Let $\mathfrak{f}:\mathbb{F}^{r}_2\times\mathbb{F}^{m}_2\rightarrow \mathbb{F}_2$ $(r,m$ even$)$ be given in the form $f:\mathbb{F}^5_2\rightarrow \mathbb{F}_2$, defined by (\ref{eq:formth22}), as $\mathfrak{f}(x,y)=f(a(x),h_1(x),h_2(x),g_1(y),g_2(y)),$ $(x,y)\in \mathbb{F}^{r}_2\times\mathbb{F}^{m}_2.$ Then:
\begin{enumerate}[i)]
\item Assume  $a\oplus \langle h_1,h_2\rangle$ is an affine space of bent functions and $g_1,g_2$ are bent. If it holds that $a^*\oplus (a\oplus h_1)^*\oplus (a\oplus h_2)^*\oplus (a\oplus h_1\oplus h_2)^*\in\{0,1\}$, then $\mathfrak{f}$ is 5-valued spectra function.
\item Assume that $a(x)$ is $4$-plateaued,  $a\oplus c_1h_1\oplus c_2h_2$  is  $2$-plateaued for $(c_1,c_2)\in \mathbb{F}^2_2\setminus \{\bf{0}_2\}$, and $g_1,g_2$ are bent functions $(r \geq6)$. If $a\oplus \langle h_1,h_2\rangle $ is an affine space of pairwise disjoint spectra plateaued functions, then $\mathfrak{f}$ is 5-valued spectra function.
\end{enumerate}
\end{theo}
\proof $i)$ Let $(u,v)\in \mathbb{F}^{r}_2\times\mathbb{F}^{m}_2$ be arbitrary. Using the fact that $a\oplus \langle h_1,h_2\rangle$ and $g_1,g_2$ are bent, then by (\ref{mainF2}) the WHT of $\mathfrak{f}$ ($s=3$) is given as
\begin{eqnarray*}
W_{\mathfrak{f}}(u,v)\hskip -2mm &=&\hskip -2mm \frac{1}{2}
\left[W_{g_2}(v)(W_{a}(u)+W_{a\oplus h_2}(u))+W_{g_1}(v)(W_{a\oplus h_1}(u)-W_{a\oplus h_1\oplus h_2}(u))\right]\\
\hskip -2mm &=&\hskip -2mm \frac{1}{2}\cdot 2^{\frac{r+m}{2}}(-1)^{g^*_2(v)}\cdot 
\widehat{\chi}\cdot \chi_{f^*},
\end{eqnarray*}
where $\chi_{f^*}=(1,1,1,-1)$ and
$$\widehat{\chi}=((-1)^{a^*(u)},(-1)^{(a\oplus h_2)^*(u)},(-1)^{(a\oplus h_1)^*(u)\oplus g^*_1(v)\oplus g^*_2(v)},(-1)^{(a\oplus h_1\oplus h_2)^*(u)\oplus g^*_1(v)\oplus g^*_2(v)}).$$
Now, it is not difficult to see that  $\widehat{\chi}$ is equal to $\pm H^{(r)}_4$ (for some $r\in[0,3]$), whenever $$a^*(u)\oplus (a\oplus h_1)^*(u)\oplus (a\oplus h_2)^*(u)\oplus (a\oplus h_1\oplus h_2)^*(u)=0,$$ and $\widehat{\chi}\neq\pm H^{(r)}_4$ when
$$a^*(u)\oplus (a\oplus h_1)^*(u)\oplus (a\oplus h_2)^*(u)\oplus (a\oplus h_1\oplus h_2)^*(u)=1.$$ In the case when $\widehat{\chi}\neq\pm H^{(r)}_4$, then $\widehat{\chi}$ is sequence of a bent function in two variables, and since $f^*$ is also bent, then $\widehat{\chi}\cdot \chi_{f^*}\in \{0,\pm 4\}$ (since the sum of any two bent functions on $\mathbb{F}^2_2$ is affine or linear). Consequently, the WHT of $\mathfrak{f}$ is given as
$$W_{\mathfrak{f}}(u,v)=\left\{\begin{array}{cc}
                                   0, & \widehat{\chi}\neq \pm H^{(r)}_4\;\;\text{and}\;\;\widehat{\chi}\cdot \chi_{f^*}=0\\
                                 \pm 2^{\frac{r+m+2}{2}}, & \widehat{\chi}\neq \pm H^{(r)}_4\;\;\text{and}\;\;\widehat{\chi}\cdot \chi_{f^*}=\pm 4 \\
                                 \pm 2^{\frac{r+m}{2}}, & \widehat{\chi}=\pm H^{(r)}_4
                               \end{array}
\right.,\;\;\;r\in[0,3],$$
so that $\mathfrak{f}$ is a 5-valued spectra function and $W_{\mathfrak{f}}(u,v)\in \{0, \pm 2^{\frac{n}{2}}, \pm 2^{\frac{n+2}{2}}\}$, as $n=r+m$.

$ii)$ Since $\#S_a=2^{r-4}$ and $\#S_{a\oplus c_1h_1\oplus c_2h_2}=2^{r-2}$ (for $(c_1,c_2)\in \mathbb{F}^2_2\setminus \{\textbf{0}_2\}$), then clearly $S_{a\oplus  c_1h_1\oplus c_2h_2}$ (where $(c_1,c_2)\in \mathbb{F}^2_2$) do not partition the space $\mathbb{F}^r_2$. Consequently,
\begin{eqnarray*}
W_{\mathfrak{f}}(u,v)= \left\{\begin{array}{cl}
                                    0, & u\not\in S_{a\oplus  c_1h_1\oplus c_2h_2},\;\;(c_1,c_2)\in \mathbb{F}^2_2, \\
                                    \frac{1}{2}W_{g_2}(v)W_{a\oplus c_1h_1\oplus c_2h_2}(u), & u\in S_{a\oplus  c_1h_1\oplus c_2h_2},\;\;(c_1,c_2)\in \{(0,0),(0,1)\}, \\
                                    \frac{1}{2}(-1)^{1\oplus c_1\oplus c_2}W_{g_1}(v)W_{a\oplus c_1h_1\oplus c_2h_2}(u), & u\in S_{a\oplus  c_1h_1\oplus c_2h_2},\;\;(c_1,c_2)\in \{(1,0),(1,1)\}.
                                  \end{array}
\right.
\end{eqnarray*}
The given assumptions imply that $W_{\mathfrak{f}}(u,v)\in \{0, \pm 2^{\frac{n}{2}}, \pm 2^{\frac{n+2}{2}}\}$, where $n=r+m$.\qed
\begin{rem}
From the proof of Theorem \ref{th:C4}-$(ii)$, one can easily determine the sets $S^{[i]}_{\mathfrak{f}}\subset \mathbb{F}^r_2\times \mathbb{F}^m_2$ and the duals $\mathfrak{f}^*_{[i]}:S^{[i]}_{\mathfrak{f}}\rightarrow \mathbb{F}_2$, for $i=1,2$. Similarly as in Theorem \ref{th:C3}-$(ii)$ one can impose conditions on functions $a,h_1,h_2,g_1$ and $g_2$ such that four functions $\mathfrak{f}_i(x,y)=f(a(x),h_1(x),h_2(x),g_1(y),g_2(y))$ (where $a,h_i,g_i$ satisfy the assumptions of Theorem \ref{th:C4}-$(ii))$ are suitable for 5-valued $4$-bent decomposition. This analysis is however omitted due to space constrains.
\end{rem}

The following example illustrates the construction of 5-valued spectra function using Theorem \ref{th:C4}-$(i)$.
\begin{ex}
Let $a,h_1,h_2,g_1,g_2:\mathbb{F}^4_2\rightarrow \mathbb{F}_2$ $(r=m=4)$ be defined as
\begin{eqnarray*}
\left\{\begin{array}{c}
         a(x_1,\ldots,x_4)=x_1x_3\oplus x_2x_4,\;\;\; h_1(x_1,\ldots,x_4)=x_3x_4,\;\;\;h_2(x_1,\ldots,x_4)=x_1x_2, \\
         g_1(x_5,\ldots,x_8)=x_5x_6\oplus x_7x_8,\;\;\;g_2(x_5,\ldots,x_8)=x_5x_7 \oplus x_6x_8.
       \end{array}
\right.
\end{eqnarray*}
One can verify that $a\oplus\langle h_1,h_2\rangle$ is an affine space of bent functions on $\mathbb{F}^4_2$, and $\zeta=a^*\oplus (a\oplus h_1)^*\oplus (a\oplus h_2)^*\oplus (a\oplus h_1\oplus h_2)^*$ is not constant on $\mathbb{F}^4_2$. More precisely, the truth table of $\zeta:\mathbb{F}^4_2\rightarrow \mathbb{F}_2$ is given as
 $$T_{\zeta}=(0, 0, 0, 0, 0, 0, 0, 0, 1, 1, 1, 1, 1, 1, 1, 1).$$
 Since $g_1,g_2$ are bent as well, then by Theorem \ref{th:C4}-$(i)$ and form (\ref{eq:formth22}) the function $\mathfrak{f}:\mathbb{F}^4\times \mathbb{F}^4_2\rightarrow \mathbb{F}_2$ is given as
\begin{eqnarray*}
\mathfrak{f}(x,y)&=&f(a(x),h_1(x),h_2(x),g_1(y),g_2(y))=a(x)\oplus g_2(y)\oplus h_2(x)(h_1(x)\oplus g_1(y)\oplus g_2(y))\\
&=&(x_1x_3 \oplus x_2x_4) \oplus (x_5x_7 \oplus x_6x_8) \oplus x_2x_3(x_3x_4 \oplus x_5x_6 \oplus x_7x_8 \oplus x_5x_7 \oplus x_6x_8).
\end{eqnarray*}
One can check that $W_{\mathfrak{f}}(u,v)\in \{0,\pm 16,\pm 32\}$, $(u,v)\in \mathbb{F}^4\times \mathbb{F}^4_2$, which means that $\mathfrak{f}$ is 5-valued function (of degree $4$).
\end{ex}

\subsection{5-valued spectra from GMM design method}\label{sec:MM}

%
A generalization of the Maiorana-McFarland design method of concatenating affine functions on a smaller variable space was recently introduces in \cite{GMM}. This approach, called GMM (Generalized Maiorana-McFarland) in \cite{GMM}, mainly concerns the design of highly nonlinear resilient functions and it provides Boolean functions that in most of the cases posses (currently) the highest known nonlinearity. Nevertheless, this approach can be adopted to give an efficient design of 5-valued spectra functions as follows.

\begin{theo}\label{GMM}
Let $E_0\subset \mathbb{F}_2^s$ with $1\leq s\leq \lfloor n/2\rfloor$. Let $E_1=\overline{E_0}\times \mathbb{F}_2^t$,
where $\overline{E_0}=\mathbb{F}_2^s\backslash E_0$ and $0\leq t\leq \lfloor n/2\rfloor$.
Let $\phi_0$ be an injective mapping from $E_0$ to $\mathbb{F}_2^{n-s}$, and $\phi_1$ be an injective mapping from $E_1$ to $\mathbb{F}_2^{n-s-t}$. Let $X=(x_1,\ldots,x_n)\in \mathbb{F}_2^n$ and $X_{(i,j)}=(x_i,\ldots,x_j)\in \mathbb{F}_2^{j-i+1}$.  $f\in \mathcal {B}_n$ is defined as follows:
\begin{equation*}\label{}
  f(X)=
  \begin{cases}
  \phi_0(X_{(1,s)})\cdot X_{(s+1,n)},& if~X_{(1,s)}\in E_0\\
  \phi_1(X_{(1,s+t)})\cdot (X_{(s+t+1,n)}),& if~X_{(1,s+t)}\in E_1.
  \end{cases}
\end{equation*}
Let
\begin{equation*}\label{}
  T_0=\{\phi_0(\eta)\mid \eta\in E_0\},
\end{equation*}
and
\begin{equation*}\label{}
  T_1=\{\phi_1(\theta)\mid \theta\in E_1\}.
\end{equation*}
Then we have

\begin{description}
  \item[a)] $W_f(\omega)\in \{0,\pm 2^{n-s},\pm 2^{n-s-t}\}$ if $t\neq 0$ and $T_0\subset \mathbb{F}_2^t\times \overline{T_1}$,
            where $\overline{T_1}=\mathbb{F}_2^{n-s-t}\backslash T_1$;
  \item[b)] $W_f(\omega)\in \{0,\pm 2^{n-s},\pm 2^{n-s+1}\}$ if $t=0$, $T_0\cap T_1\neq \emptyset$ and $T_0\neq T_1$.
\end{description}
\end{theo}

\begin{proof}
Let $\omega=(\omega_1,\cdots,\omega_n)\in \mathbb{F}_2^n$. We have
\begin{align*}
   W_f(\omega)= \sum_{X\in \mathbb{F}_2^n}(-1)^{f(X)\oplus \omega\cdot X } =S_1(\omega)+S_2(\omega),
\end{align*}
where
\begin{equation*}\label{}
  S_1(\omega)=\sum_{X_{(1,s)}\in E_0}(-1)^{\omega_{(1,s)}\cdot X_{(1,s)}}\sum_{X_{(s+1,n)}\in \mathbb{F}_2^{n-s}}(-1)^{\phi_0(X_{(1,s)})\oplus \omega_{(s+1,n)} \cdot X_{(s+1,n)}}
\end{equation*}
and
\begin{equation*}\label{}
  S_2(\omega)=\sum_{X_{(1,s+t)}\in E_1}(-1)^{\omega_{(1,s+t)}\cdot X_{(1,s+t)}}\sum_{X_{(s+t+1,n)}\in \mathbb{F}_2^{n-s-t}}(-1)^{(\phi_0(X_{(1,s+t)})\oplus \omega_{(s+t+1,n)})\cdot X_{(s+t+1,n)}}.
\end{equation*}
Furthermore, we have
\begin{eqnarray*}\label{s1}
 S_1(\omega)=\left\{ \begin{array}{ll}
 \pm 2^{n-s}, & \textrm{if $\phi_0^{-1}(\omega_{(s+1,n)})$ exists}\\
 0, & \textrm{otherwise}
\end{array} \right.
\end{eqnarray*}
and
\begin{eqnarray*}\label{s2}
 S_2(\omega)=\left\{ \begin{array}{ll}
 \pm 2^{n-s-t}, & \textrm{if $\phi_1^{-1}(\omega_{(s+t+1,n)})$ exists}\\
 0, & \textrm{otherwise}.
\end{array} \right.
\end{eqnarray*}
By the condition a), for any $\omega\in \mathbb{F}_2^n$, we always have
\begin{equation*}\label{}
  S_1(\omega)\cdot S_2(\omega)=0,
\end{equation*}
which implies $S_1(\omega)+ S_2(\omega)\in \{0,\pm 2^{n-s},\pm 2^{n-s-t}\}$.

By condition b), there exists a $\omega_{(s+1,n)}$ such that $\phi_0^{-1}(\omega_{(s+t+1,n)})$ and $\phi_1^{-1}(\omega_{(s+1,n)})$ both exist, and there also exist a $\omega_{(s+1,n)}$ such that $  S_1(\omega)\cdot S_2(\omega)=0$. This implies that $S_1(\omega)+ S_2(\omega)\in \{0,\pm 2^{n-s},\pm 2^{n-s+1}\}$. \qed
\end{proof}
Whereas throughout this article we were mainly concerned with design methods (in spectral or ANF domain) some additional properties of the constructed functions have not been discussed. However, the above method (stemming from the GMM design method for resilient functions)  intrinsically embeds  a certain order of resiliency into the designed functions. Denoting by $ m_0=\min\{wt(\alpha)\mid \alpha\in T_0\}$ and $ m_1=\min\{wt(\beta)\mid \beta\in T_1\}$ one can easily verify that $f$ is an $m$-resilient functions,  where $m=\min\{m_0,m_1\}$ \cite{XiaoMassey}.  When $n$ is odd and $s = \lfloor n/2\rfloor =(n-1)/2$, this design yields functions with cryptographically interesting spectra (when $t=1$) of the form $W_f(\omega)\in \{0,\pm 2^{\frac{n-1}{2}},\pm 2^{\frac{n+1}{2}}\}$. 
Thus, resilient Boolean functions with high nonlinearity can be generated using this approach even though the GMM class in general attains higher nonlinearities.


\section{Conclusions}\label{sec:concl}
In this article we have provided both spectral and ANF design of 5-valued spectra functions which are of significant theoretical and practical interest. Apart form providing an interest class of Boolean functions with applications in cryptography, these functions may  also be used to construct bent functions on $\FB^n$ (by concatenating four suitable 5-valued spectra functions) and furthermore the Walsh spectra of Dillon's APN permutations on $\FB$ are 5-valued which may encourage for further analysis in this direction using the proposed generic methods. The question regarding affine equivalence of the designed classes has not been considered here but we believe that this issue is an interesting research topic.
\\\\
\noindent
{\large \bf Acknowledgment:} Samir Hod\v zi\' c  is supported in part by the Slovenian Research Agency (research program P3-0384 and Young Researchers Grant). Enes Pasalic is partly supported  by the Slovenian Research Agency (research program P3-0384 and research project J1-9108). For the first two authors, the work is supported in part by H2020 Teaming InnoRenew CoE (grant no. 739574). WeiGuo Zhang is supported by the National Natural Science Foundation of China (grant no. 61672414), and the National Cryptography Development Fund (Grant no. MMJJ20170113).


\begin{thebibliography}{10}



\bibitem{Decom}
{\sc A. Canteaut, P. Charpin.}
\newblock Decomposing bent functions.
\newblock {\em IEEE Transactions on Information Theory}, vol. 49, no. 8, pp. 2004--2019, 2003.

\bibitem{Pascale}
{\sc A. Canteaut, C. Carlet, P. Charpin, C. Fontaine.}
\newblock On cryptographic properties of the cosets of $R(1, m)$.
\newblock {\em IEEE Transactions on Information Theory}, vol. 47, no. 4, pp. 1494 -- 1513, 2001.

\bibitem{CaoXu}
{\sc X. Cao, L. Xu.}
\newblock  Two Boolean functions with five-valued Walsh spectra and high nonlinearity.
\newblock \emph{International Journal of Foundations of Computer Science}, vol. 26, no. 5, pp. 537--556, 2015.

\bibitem{CarletBoolean}
{\sc C. Carlet.}
\newblock Boolean models and methods in mathematics, computer science, and engineering.
\newblock {\em  Encyclopedia of Mathematics and its Applications (No. 134) - Cambridge University Press}, pp. 398 -- 469, 2013. 

\bibitem{CarletAPN}
{\sc C. Carlet.}
\newblock Boolean and vectorial plateaued functions and APN functions.
\newblock {\em IEEE Transactions on Information Theory}, vol. 61, no. 11, pp. 6272--6289, 2015.

\bibitem{CarlMes2016}
{\sc C. Carlet,   S. Mesnager.}
\newblock  Four decades of research on bent functions.
\newblock {\em Designs,  Codes and Cryptography}, vol. 78, no. 1, pp. 5--50, 2016.

\bibitem{CarletQ}
{\sc C. Carlet, E. Prouff.}
\newblock On plateaued functions and their constructions.
\newblock {\em  International Workshop on Fast Software Encryption, FSE 2003}, pp. 54--73, 2003.

\bibitem{THOMCusick2016}
{\sc T. W. Cusick.}
\newblock Highly nonlinear plateaued functions.
\newblock {\em IET Information Security}, vol. 11, no. 2, pp. 78--81, 2016.

\bibitem{Nastja}
{\sc N. Cepak, E. Pasalic, A. Muratovi\'{c}-Ribi\'{c}.}
\newblock Frobenius linear translators giving rise to new infinite classes of permutations and bent functions.
\newblock {\em } Available at: http://arxiv-export-lb.library.cornell.edu/abs/1801.08460

\bibitem{Dillon}
{\sc J. F. Dillon.}
\newblock APN polynomials: An update.
\newblock {\em In Fq9, the 9th International Conference on Finite Fields and Applications}, 2009.

\bibitem{SHCF}
{\sc S. Hod\v zi\'c, E. Pasalic, Y. Wei.}
\newblock  A general framework for  secondary constructions of bent and plateaued functions.
\newblock \emph{IEEE Transactions on Information Theory}, Under revision. Available at https://arxiv.org/abs/1809.07390


\bibitem{Maitra}
{\sc S. Maitra, P. Sarkar.}
\newblock Cryptographically significant Boolean functions with  five valued Walsh spectra.
\newblock {\em Transactions on Computer Science}, vol. 276, no. 1-2, pp. 133--146, 2002.

\bibitem{SHEA}
{\sc S. Hod\v zi\' c, E. Pasalic, Y. Wei, F. Zhang.}
\newblock  Designing plateaued Boolean functions in spectral domain and their classification.
\newblock Available at: https://arxiv.org/pdf/1811.04171.pdf

\bibitem{Pren}
{\sc B. Preneel, W. V. Leekwijck, L. V. Linden, R. Govaerts, J. Vandewalle.}
\newblock Propagation characteristics of Boolean functions.
\newblock {\em  Advances in Cryptology EUROCRYPT'90, Springer, Berlin, Heidelberg}, LNCS vol. 473, pp. 161--173, 1991.

\bibitem{SihemPOP2014}
{\sc S. Mesnager.}
\newblock On semi-bent functions and related plateaued functions over the Galois field $\F_{2^n}$.
\newblock In Proceedings,  {\em Open Problems in Mathematics and Computational Science}, LNCS, Springer, pp. 243--273, 2014.

\bibitem{SihemM2017AMC}
{\sc S. Mesnager,  F. Zhang.}
\newblock On constructions of bent, semi-bent and five valued spectrum functions from old bent functions.
\newblock {\em Advances in Mathematics of Communications} vol. 11, no. 2, pp. 339--345, 2017.


\bibitem{Rot}
{\sc O. S. Rothaus.}
\newblock On "bent" functions.
\newblock {\em Journal of Combinatorial Theory, Series A}, vol. 20, no. 3, pp. 300--305, 1976.



\bibitem{Tokareva}
{\sc N. Tokareva.}
\newblock On the number of bent functions from iterative constructions: lower bounds and hypotheses.
\newblock {\em Advances in Mathematics of Communications}, vol. 5, no. 4, pp. 609--621, 2011.


\bibitem{XiaoMassey}
{\sc G-Z. Xiao, J. L. Massey.}
\newblock A spectral characterization of correlation-immune combining functions.
\newblock{\em IEEE Transactions on Information Theory}, vol. 34, no. 3, pp. 569--571, 1988.


\bibitem{XuCao}
{\sc G. Xu, X. Cao, S. Xu.}
\newblock Several classes of Boolean functions with few Walsh transform values.
\newblock \emph{Applicable Algebra in Engineering, Communication and Computing}, vol. 28, no. 2, pp. 155--176, 2017.


\bibitem{DisjSpectra}
{\sc Y. Wei, E. Pasalic, F. Zhang, W. Wu, C-X. Wang.}
\newblock  New constructions of resilient functions with strictly almost optimal nonlinearity via non-overlap spectra functions.
\newblock \emph{Information Sciences}, vol. 415--416, pp. 377--396, 2017.

\bibitem{Feng2}
{\sc F. Zhang, C. Carlet, Y. Hu, T. J. Cao.}
\newblock Secondary constructions of highly nonlinear Boolean functions and disjoint spectra plateaued functions.
\newblock {\em  Information Sciences}, vol. 283, pp. 94--106, 2014.

\bibitem{Fengrong2018IT}
{\sc F. Zhang, Y. Wei, E. Pasalic, S. Xia.}
\newblock Large sets of disjoint spectra plateaued functions inequivalent to partially linear functions.
\newblock {\em  IEEE Transactions on Information Theory}, vol. 64, no. 4, pp. 2987--2999, 2018.


\bibitem{GMM}
{\sc W-G. Zhang, E. Pasalic.}
\newblock Generalized Maiorana-McFarland construction of resilient Boolean functions with high nonlinearity and good algebraic properties.
\newblock {\em IEEE Transactions on Information Theory}, vol. 60, no. 10, pp. 6681--6695, 2014.

\bibitem{Zheng}
{\sc Y. Zheng, X. M. Zhang.}
\newblock On plateaued functions.
\newblock {\em IEEE Transactions on Information Theory}, vol. 47, no. 3, pp. 1215--1223, 2001.


\end{thebibliography}
\end{document}